\documentclass{article}
\newif\ifacmstyle
\acmstylefalse
\newif\ifappendix
\appendixtrue
\newif\ifbibtex
\ifappendix
\bibtextrue
\else
\bibtexfalse
\fi

\ifacmstyle
\fi
\usepackage{lmodern} % More complete than the default CM font
\usepackage[T1]{fontenc} % needed for lmodern
\usepackage[utf8]{inputenc}
\usepackage[polutonikogreek,english]{babel}
\usepackage{bm}
\usepackage{silence} % Shut up LaTeX warnings selectively using \WarningFilter
\ifacmstyle
\usepackage[pdfpagelabels=false]{hyperref} % Must be loaded before cleveref. Option is to shut up hyperref warning in ACM style
\else
\usepackage{hyperref} % Must be loaded before cleveref
\fi
\ifacmstyle
\else
\usepackage{fullpage}
\fi
\usepackage[capitalise]{cleveref} % to have smart references using \cref
\crefname{section}{Sect.~}{Sect.~}
\Crefname{section}{Sect.~}{Sect.}
\crefname{subsection}{Sect.~}{Sect.~}
\Crefname{subsection}{Sect.~}{Sect.}
\crefname{subsubsection}{Sect.~}{Sect.~}
\Crefname{subsubsection}{Sect.~}{Sect.}
\crefformat{section}{Sect.~#2#1#3}
\crefformat{subsection}{Sect.~#2#1#3}
\crefformat{subsubsection}{Sect.~#2#1#3}
\crefname{theorem}{thm.}{thms.}
\Crefname{theorem}{Thm.}{Thms.}
\crefformat{theorem}{Thm.~#2#1#3}
\crefname{figure}{fig.}{figs.}
\Crefname{figure}{Fig.}{Figs.}
\crefformat{figure}{Fig.~#2#1#3}
\usepackage{dsfont} % for indicator function (XXX check if needed)
\ifappendix
\else
\usepackage{enumitem} % compressed itemize, enumerate, description
\setlist[itemize]{label=${\bullet}$,leftmargin=1em,topsep=0pt,parsep=0pt,partopsep=0pt,itemsep=0pt}
\setlist[enumerate]{label={\arabic*.},leftmargin=1em,topsep=0pt,parsep=0pt,partopsep=0pt,itemsep=0pt}
\setlist[description]{leftmargin=1em,topsep=0pt,parsep=0pt,partopsep=0pt,itemsep=0pt}
\fi
\usepackage{booktabs}
\usepackage{multirow}
\usepackage{multicol}
\usepackage[numbers,sort&compress,square]{natbib}
\ifacmstyle

\fi
\usepackage[ruled,linesnumbered]{algorithm2e} % must be loaded after natbib
\usepackage{tkz-graph} % for drawing graphs.
\usepackage{ctable} % this must be loaded after tkz
\usepackage{url} % XXX check if needed
\usepackage{xspace} % used for space after commands
\ifappendix
\usepackage{subcaption} % for subfigures and such.
\fi

\ifacmstyle
\else
\usepackage{mathtools}
\usepackage{amsfonts}
\usepackage{amsthm}
\fi
\newtheorem{corollary}{Corollary}
\newtheorem{lemma}{Lemma}
\newtheorem{theorem}{Theorem}

\newtheorem{conjecture}{Conjecture}
\ifacmstyle
\else
\theoremstyle{definition}
\fi
\newtheorem{definition}{Definition}

\newcommand*\betw{\mathsf{b}}
\newcommand*\tbetw{\widetilde{\betw}}
\newcommand*\domain{\mathcal{D}}
\newcommand*\range{\mathcal{R}}
\newcommand*\SP{\mathsf{S}}
\newcommand*\VC{\mathsf{VC}}
\newcommand*\EVC{\mathsf{E}\VC}
\newcommand*\PD{\mathsf{PD}}
\newcommand*\EPD{\mathsf{E}\PD}
\newcommand*\rade{\mathsf{R}}
\newcommand*\expectation{\mathbb{E}}

\newcommand*\mean{\mathsf{m}}

\newcommand*\Sam{\mathcal{S}}
\newcommand*\family{\mathcal{F}}
\newcommand*\functionw{\mathsf{w}}

\newcommand*\BC{\textsc{bc}\xspace}
\newcommand*\BA{\textsf{BA}\xspace}
\newcommand*\algobase{ABRA}
\newcommand*\staticalgo{\textsf{\algobase-s}\xspace}
\newcommand*\staticalgotopk{\textsf{\algobase-k}\xspace}
\newcommand*\staticalgorel{\textsf{\algobase-r}\xspace}
\newcommand*\dynamicalgo{\textsf{\algobase-d}\xspace}

\newcommand*\TOP{\mathsf{TOP}}
\newcommand{\greek}[1]{{\selectlanguage{polutonikogreek}#1}}

\newcommand{\para}[1]{\noindent{\bf #1.}}

\begin{document}
\title{\algobase: Approximating Betweenness Centrality\\in Static and Dynamic Graphs with
Rademacher Averages\footnote{This work was supported in part by NSF grant
IIS-1247581 and NIH grant R01-CA180776.}}

\author{Matteo
Riondato\footnote{Two Sigma Investments. Part of the work done while
	affiliated to Brown University. \url{matteo@twosigma.com}} \and
	Eli Upfal\footnote{Department of Computer Science, Brown University.
	\url{eli@cs.brown.edu}}
}

\date{\today}

\maketitle

\noindent{\em\greek{ABRAXAS} (ABRAXAS): Gnostic word of mystic meaning}
\ifappendix
\else
\vspace{-20pt}
\fi
\begin{abstract}
	%Betweenness Centrality (\BC) is a fundamental measure in graph analysis. The
	%\BC of a vertex (or edge) $z$ is, informally, the fraction of all
	%shortest paths in the graph that go through $z$. Computing the exact \BC of
	%every node (or edge) is extremely expensive on large static graphs. Keeping
	%the \BC values up-to-date in a dynamic graph, where edge insertions and
	%deletions are allowed, is even more challenging.
	We present \textsf{\algobase}, a suite of algorithms that compute and
	maintain probabilistically-guaranteed, high-quality, approximations of the
	betweenness centrality of all nodes (or edges) on both static and fully
	dynamic graphs. Our algorithms rely on random sampling and their analysis
	leverages on Rademacher averages and pseudodimension, fundamental concepts
	from statistical learning theory.  To our knowledge, this is the first
	application of these concepts to the field of graph analysis. The results of
	our experimental evaluation show that our approach is much faster than
	exact methods, and vastly outperforms, in both speed and number of samples,
	current state-of-the-art algorithms with the same quality guarantees.
\end{abstract}

\section{Introduction}\label{sec:intro}
Centrality measures are fundamental concepts in graph analysis, as they assign
to each node or edge in the network a score that quantifies some notion of
importance of the node/edge in the network~\citep{Newman10}. Betweenness
Centrality (\BC) is a very popular centrality measure that, informally, defines
the importance of a node or edge $z$ in the network as proportional to
the fraction of shortest paths in the network that go through
$z$~\citep{Anthonisse71,Freeman77}.

\citet{Brandes01} presented an algorithm (denoted \BA) that computes the exact
\BC values for all nodes or edges in a graph $G=(V,E)$ in time $O(|V||E|)$ if
the graph is unweighted, and time $O(|V||E|+|V|^2\log|V|)$ if the graph has
positive weights. The cost of \BA is excessive on modern networks with millions
of nodes and tens of millions of edges. Moreover, having the exact \BC values
may often not be needed, given the exploratory nature of the task, and a
high-quality approximation of the values is usually sufficient, provided it
comes with stringent guarantees.

Today's networks are not only large, but also \emph{dynamic}: edges are
added and removed continuously. Keeping the \BC values up-to-date after edge
insertions and removals is a challenging task, and proposed
algorithms~\citep{KourtellisMB15,GreenMB12,KasWCC13,LeeLPCC21} have a worst-case
complexity and memory requirements which is not better than
from-scratch-recomputation using \BA. Maintaining an high-quality approximation
up-to-date is more feasible and more \emph{sensible}: there is little added
value in keeping track of exact \BC values that change continuously.

\paragraph*{Contributions} We focus on developing algorithms for approximating
the \BC of all vertices and edges in static and dynamic graphs. Our
contributions are the following.

\begin{itemize}
	\item We present \textsf{\algobase} (for ``\textsf{A}pproximating
		\textsf{B}etweenness with \textsf{R}ademacher \textsf{A}verages''), the
		first family of algorithms based on \emph{progressive sampling} for
		approximating the \BC of all vertices in static and dynamic graphs,
		where vertex and edge insertions and deletions are allowed. The
		approximations computed by \textsf{\algobase{}} are
		\emph{probabilistically guaranteed} to be within an user-specified
		additive error from their exact values. We also present variants with
		relative (i.e., multiplicative)) error for the top-$k$ vertices with
		highest \BC, and variants that use refined estimators to give better
		approximations %in practice
		with a slightly larger sample size.
	\item Our analysis relies on Rademacher averages~\citep{ShalevSBD14} and
		pseudodimension~\citep{Pollard84}, fundamental concepts from the field
		of statistically learning theory~\citep{Vapnik99}. Exploiting known and
		novel results using these concepts, \textsf{\algobase{}} computes the
		approximations without having to keep track of any global property of
		the graph, in contrast with existing
		algorithms~\citep{RiondatoK15,BergaminiM15,BergaminiMS15}.
		\textsf{\algobase{}} performs only ``real work'' towards
		the computation of the approximations, without having to compute such
		global properties or update them after modifications of the graph. To
		the best of our knowledge, ours is the first application of Rademacher
		averages and pseudodimension to graph analysis problems, and the first
		to use \emph{progressive} random sampling for \BC computation. Using
		pseudodimension new analytical results on the sample complexity of the
		\BC computation task, generalizing previous
		contributions~\citep{RiondatoK15}, and formulating a conjecture on the
		connection between pseudodimension and the distribution of shortest path
		lengths.
	\item The results of our experimental evaluation on real % and synthetic
		networks show that \textsf{\algobase{}} outperforms, in both speed and
		number of samples, the state-of-the-art methods offering the same
		guarantees~\citep{RiondatoK15}.
\end{itemize}

\ifappendix
\paragraph*{Outline} We discuss related works in~\cref{sec:related}. The formal
definitions of the concepts we use in the work can be found
in~\cref{sec:prelims}. Our algorithms for approximating \BC on static graphs are
presented in~\cref{sec:static}, while the dynamic case is discussed
in~\cref{sec:dynamic}. The results of our extensive experimental evaluation are
presented in~\cref{sec:exper}. We draw conclusions and outline directions for
future work in~\cref{sec:concl}. Additional details can be found in the
Appendices.
\else
Due to space constraints, some details have been deferred to the extended online
version~\citep{RiondatoU16ext}.
\fi

\section{Related Work}\label{sec:related}
The definition of Betweenness Centrality comes from the sociology
literature~\citep{Anthonisse71,Freeman77}, but the study of efficient algorithms
to compute it started only when graphs of substantial size became available to
the analysts, following the emergence of the Web. The \BA algorithm
by~\citet{Brandes01} is currently the asymptotically fastest algorithm for
computing the exact \BC values for all nodes in the network. A number of works
also explored heuristics to improve \BA~\citep{SaryuceSKC13,ErdosIBT15}, but
retained the same worst-case time complexity.

The use of random sampling to approximate the \BC values in static graphs was
proposed independently by~\citet{BaderKMM07} and~\citet{BrandesP07}, and
successive works explored the tradeoff space of sampling-based
algorithms~\citep{RiondatoK15,BergaminiM15,BergaminiMS15,BergaminiM15arXiv}.
We focus here on related works that offer approximation guarantees similar to
ours. For an in-depth discussion of previous contributions approximating \BC on
static graphs, we refer the reader to~\citep[Sect.~2]{RiondatoK15}.

\citet{RiondatoK15} present algorithms that employ the Vapnik-Chervonenkis (VC)
dimension~\citep{Vapnik99} to compute what is currently the tightest upper
bound to the sample size sufficient to obtain guaranteed approximations of the
\BC of all nodes in a static graph. Their algorithms offer the same guarantees
as ours, but they need to compute an upper bound to a characteristic quantity of
the graph (the vertex diameter, namely the maximum number of nodes on any
shortest path) in order to derive the sample size. Thanks to our use of
Rademacher averages in a progressive random sampling setting, we do not need to
compute any characteristic quantity of the graph, and instead use an
efficient-to-evaluate stopping condition to determine when the approximated \BC
values are close to the exact ones. This allows \textsf{\algobase{}} to use
smaller samples and be much faster than the algorithm by~\citet{RiondatoK15}.

A number of works~\citep{KourtellisMB15,GreenMB12,KasWCC13,LeeLPCC21} focused on
computing the \emph{exact} \BC for all nodes in a dynamic graph, taking into
consideration different update models. None of these algorithm is provably
asymptotically faster than a complete computation from scratch using Brandes'
algorithm~\citep{Brandes01} and they all require significant amount of space
(more details about these works can be found in~\citep[Sect.~2]{BergaminiM15}).
In contrast, \citet{BergaminiM15,BergaminiM15arXiv} built on the work
by~\citet{RiondatoK15} to derive an algorithm for maintaining high-quality
approximations of the \BC of all nodes when the graph is dynamic and both
additions and deletions of edges are allowed. Due to the use of the
algorithm by~\citet{RiondatoK15} as a building block, the algorithm must keep
track of the vertex diameter after an update to the graph. Our algorithm for
dynamic graphs, instead, does not need this piece of information, and therefore
can spend more time in computing the approximations, rather than in keeping
track of global properties of the graph. Moreover, our algorithm can handle
directed graphs, which is not the case for the algorithms
by~\citet{BergaminiM15,BergaminiM15arXiv}.

\citet{HayashiAY15} recently proposed a data structure called \emph{Hypergraph
Sketch} to maintain the shortest path DAGs between pairs of nodes
following updates to the graph. Their algorithm uses random sampling and this
novel data structure allows them to maintain a high-quality, probabilistically
guaranteed approximation of the \BC of all nodes in a dynamic graph. Their
guarantees come from an application of the simple uniform deviation bounds
(i.e., the union bound) to determine the sample size, as previously done
by~\citet{BaderKMM07} and~\citet{BrandesP07}. As a result, the resulting sample
size is excessively large, as it depends on the \emph{number of nodes in the
graph}. Our improved analysis using the Rademacher averages allows us to develop
an algorithm that uses the Hypergraph Sketch with a much smaller number of
samples, and is therefore faster. %by~\citet{HayashiAY15}.

\section{Preliminaries}\label{sec:prelims}
We now introduce the formal definitions and basic results that we use throughout
the paper.

\subsection{Graphs and Betweenness Centrality}
Let $G=(V,E)$ be a graph, which can be directed or undirected, and can have
non-negative weights on the edges. For any ordered pair $(u,v)$ of different
nodes $u\neq v$, let $\mathcal{S}_{uv}$ be the set of \emph{Shortest Paths}
(SPs) from $u$ to $v$, and let $\sigma_{uv}=|\mathcal{S}_{uv}|$. Given a path
$p$ between two nodes $u,v\in V$, a node $w\in V$ is \emph{internal to $p$} iff
$w\neq u$, $w\neq u$, and $p$ goes through $w$. We denote as $\sigma_{uv}(w)$
the number of SPs from $u$ to $v$ that $w$ is internal to.

\begin{definition}[\citep{Anthonisse71,Freeman77}]\label{def:betweenness}
	Given a graph $G=(V,E)$, the \emph{Betweenness Centrality (\BC) of a vertex $w\in
	V$} is defined as
	\[
		\betw(w)=\frac{1}{|V|(|V|-1)}\sum_{\substack{(u,v)\in V\times V\\u\neq
		v}}\frac{\sigma_{uv}(w)}{\sigma_{uv}}\enspace.
	\]
\end{definition}
We have $\betw(w)\in[0,1]$, for any $w\in V$. Many variants of \BC have been
proposed in the literature, including one for edges~\citep{Newman10}. All our
results can be extended to these variants, following the reduction
in~\citep[Sect.~6]{RiondatoK15}, but we do not include them here due to space
constraints.

In this work we focus on computing an \emph{$(\varepsilon,\delta)$-approximation}
of the collection $B=\{\betw(w), w\in V\}$.

\begin{definition}\label{def:eapprox}
	Given $\varepsilon,\delta\in(0,1)$, an
	\emph{$(\varepsilon,\delta)$-approximation to $B$} is a collection
	$\tilde{B}=\{\tilde{\betw}(w), w\in V\}$ such that
	\[
			\Pr(\forall w\in v ~:~ |\tilde{\betw}(w) -
			\betw(w)|\leq \varepsilon)\ge 1-\delta\enspace.
	\]
\end{definition}

%The (asymptotically) fastest known algorithm for computing the exact \BC of all
%nodes in the graph was proposed by~\citet{Brandes01}, and takes time $O(|V||E|)$
%on unweighted graphs, and $O(|V||E|+|V|^2\log |V|)$ on weighted networks.

\subsection{Rademacher Averages}\label{sec:radeprelims}
Rademacher Averages are fundamental concepts to study the rate of convergence of
a set of sample averages to their expectations. They are at the core of
statistical learning theory~\citep{Vapnik99} but their usefulness extends way
beyond the learning framework~\citep{RiondatoU15}. We present here only the
definitions and results that we use in our work and we refer the readers to,
e.g., the book by~\citet{ShalevSBD14} for in-depth presentation and discussion.

While the Rademacher complexity can be defined on an arbitrary measure space, we restrict our discussion here to a sample space
that consists of a finite domain $\domain$ and a uniform distribution over that domain.
Let $\family$ be a family of functions from $\domain$ to $[0,1]$,  and
let $\Sam=\{c_1,\dotsc,c_\ell\}$ be a sample of $\ell$
elements from $\domain$, sampled uniformly and independently at random.
For each $f\in\family$, the \emph{true sample} and the \emph{sample average} of $f$ on a sample
$\Sam$ are
\begin{equation}\label{eq:mean}
	\mean_\domain(f) \frac{1}{|\domain|}\sum_{c\in \domain}^\ell f(c)  ~~\mbox{and}~~             \mean_\Sam(f)=\frac{1}{\ell}\sum_{i=1}^\ell f(c_i).
	\end{equation}
Given $\Sam$, we are interested in bounding the \emph{maximum deviation of
	$\mean_\Sam(f)$ from $\mean_\domain(f)$}, i.e., in the quantity
\begin{equation}\label{eq:supdev}
	\sup_{f\in\family}|\mean_\Sam(f)-\mean_\domain(f)|\enspace.
\end{equation}
For $1\le i\le \ell$, let $\sigma_i$ be a Rademacher r.v., i.e., a
r.v.~that takes value $1$ with probability $1/2$ and $-1$ with
probability $1/2$. The r.v.'s~$\sigma_i$ are independent. Consider the quantity
\begin{equation}\label{eq:rade}
	\rade(\family,\Sam)=\expectation_\sigma\left[\sup_{f\in
	\family}\frac{1}{\ell}\sum_{i=1}^\ell\sigma_if(c_i)\right],
\end{equation}
where the expectation is taken w.r.t.~the Rademacher r.v.'s, i.e.,
conditionally on $\Sam$. The quantity $\rade(\family,\Sam)$ is known as the
\emph{(conditional) Rademacher average of $\family$ on $\Sam$}. The following
is a key result in statistical learning theory, connecting $\rade(\family,\Sam)$
to the maximum deviation~\eqref{eq:supdev}.

\begin{theorem}[Thm.~26.5~\citep{ShalevSBD14}]\label{thm:supdevbound}
	Let $\delta\in(0,1)$ and let $\Sam$ be a collection of $\ell$ elements of
	$\domain$ sampled independently and uniformly at random. Then, with
	probability at least $1-\delta$,
	\begin{equation}\label{eq:supdevbound}
		\sup_{f\in\family}|\mean_\Sam(f)-\mean_\domain(f)|\le
		2\rade(\family,\Sam) + 3\sqrt{\frac{\ln(2/\delta)}{2\ell}}\enspace.
	\end{equation}
\end{theorem}
\Cref{thm:supdevbound} is how the result is classically presented, but better
although more complex bounds than~\eqref{eq:supdevbound} are available~\citep{OnetoGAR13}.

\begin{theorem}[Thm.~3.11~\citep{OnetoGAR13}]\label{thm:supdevboundrefined}
	Let $\delta\in(0,1)$ and let $\Sam$ be a collection of $\ell$ elements of
	$\domain$ sampled independently and uniformly at random. Let
	\begin{equation}\label{eq:alpha}
		\alpha=\frac{\ln\frac{2}{d}}{\ln\frac{2}{d}+\sqrt{\left(2\ell\rade(\family,\Sam)+\ln\frac{2}{d}\right)\ln\frac{2}{d}}},
	\end{equation}
	then, with probability at least $1-\delta$,
	\begin{equation}\label{eq:supdevboundrefined}
		\sup_{f\in\family}|\mean_\Sam(f)-\mean_\domain(f)|\le\frac{\rade(\family,\Sam)}{1-\alpha}+\frac{\ln\frac{2}{d}}{2\ell\alpha(1-\alpha)}+\sqrt{\frac{\ln\frac{2}{d}}{2\ell}}\enspace.
	\end{equation}
\end{theorem}

Computing, or even estimating, the expectation in~\eqref{eq:rade} w.r.t.~the
Rademacher r.v.'s is not straightforward, and can be computationally expensive,
requiring a time-consuming Monte Carlo simulation~\citep{BoucheronBL05}. For
this reason, \emph{upper bounds to the Rademacher average} are usually employed
in~\eqref{eq:supdevbound} and~\eqref{eq:supdevboundrefined} in place of
$\rade(\family,\Sam)$. A powerful and efficient-to-compute bound is presented
in~\Cref{thm:radeboundw}. Given $\Sam$, consider, for each $f\in\family$, the
vector
$\mathbf{v}_{f,\Sam}=(f(c_1),\dotsc,f(c_\ell))$, and let
$\mathcal{V}_\Sam=\{\mathbf{v}_f, f\in\family\}$ be the \emph{set} of such
vectors ($|\mathcal{V}_\Sam|\le|\family|$).

\begin{theorem}[\ifappendix Thm.~3~\else\citep{RiondatoU15}\fi]\label{thm:radeboundw}
  Let $\functionw: \mathbb{R}^+ \to \mathbb{R}^+$ be the function
  \begin{equation}\label{eq:functionw}
	\functionw(s) =
	\frac{1}{s}\ln\displaystyle\sum_{\mathbf{v}\in
	\mathcal{V}_\Sam}\mathrm{exp}(s^2\|\mathbf{v}\|^2/(2\ell^2)),
  \end{equation}
  where $\|\cdot\|$ denotes the Euclidean norm. Then
  \begin{equation}\label{eq:functionwbound}
	\rade(\family,\Sam)\leq\min_{s\in\mathbb{R}^+}\functionw(s)\enspace.
\end{equation}
\end{theorem}
The function $\functionw$ is convex, continuous in $\mathbb{R}^+$, and has first
and second derivatives w.r.t.~$s$ everywhere in its domain, so it is possible to
minimize it efficiently using standard convex optimization
methods~\citep{BoydV04}. In future work, we plan to explore how to obtain a
tighter bound than the one presented
in~\Cref{thm:radeboundw} using recent results by~\citet{AnguitaGOR14}.
%\mynote{I actually have no idea of what I meant by that.}

\section{Static Graph BC Approximation}\label{sec:static}
We now present and analyze \staticalgo, our \emph{progressive sampling
algorithm} for computing an $(\varepsilon,\delta)$-approximation to the
collection of exact \BC values in a static graph. Many of the details and
properties of \staticalgo are shared with the other \textsf{\algobase{}}
algorithms we present.

\para{Progressive Sampling} Progressive sampling algorithms are intrinsically
\emph{iterative}. At a high level, they work as follows. At iteration $i$, the
algorithm extracts an approximation of the values of interest (in our case, of
the \BC of all nodes) from a collection $\Sam_i$ of $S_i$ random samples from a
suitable domain $\domain$ (in our case, the samples are pairs of different
nodes). Then, the algorithm checks a specific \emph{stopping condition}
which uses information obtained from the sample and from the computed
approximation. If the stopping condition is satisfied, then the approximation
has, with the required probability, the desired quality (in our case, it is an
$(\varepsilon,\delta)$-approximation), and can be returned in output,
at which point the algorithm terminates. If the stopping condition is not
satisfied, the algorithm builds a collection $\Sam_{i+1}$ by adding random
samples to the $\Sam_i$ until $S_{i+1}$, the algorithm iterates, computing a new
approximation from the so-created collection $\Sam_{i+1}$.

There are two main challenges for the algorithm designer: deriving a ``good''
stopping condition and determining the initial sample size $S_1$ and the next
sample sizes $S_{i+1}$.

Ideally, one would like a stopping condition that:
\begin{enumerate}
	\item when satisfied, guarantees that the computed approximation has the
		desired quality properties (in our case, it is an
			$(\varepsilon,\delta)$-approximation; and
	\item can be evaluated efficiently; and
	\item is tight, in the sense that is satisfied at small sample sizes.
\end{enumerate}
The stopping condition for our algorithm is based on~\Cref{thm:radeboundw}
and~\Cref{thm:supdevboundrefined} and has all the above desirable properties.

The second challenge is determining the \emph{sample schedule $(S_i)_{i>0}$}.
Any monotonically increasing sequence of positive numbers can act as sample
schedule, but the goal in designing a good sample schedule is to minimize the
number of iterations that are needed before the stopping condition is satisfied,
while minimizing the sample size $S_i$ at the
iteration $i$ at which this happens. %the stopping condition is satisfied.
The sample schedule may be fixed in advance, but %it should be clear from this discussion that
an \emph{adaptive approach} that ties the sample schedule to the
stopping condition can give better results, as the sample size
$S_{i+1}$ for iteration $i+1$ can be computed using information obtained in (or
up-to) iteration $i$. \textsf{\algobase{}} uses such an adaptive approach.

\subsection{Algorithm Description and Analysis}\label{sec:description}
\staticalgo takes as input a graph $G=(V,E)$ and two
parameters $\varepsilon,\delta\in(0,1)$, and outputs a collection
$\widetilde{B}=\{\tbetw(w), w\in V\}$ that is an
$(\varepsilon,\delta)$-approximation of the betweenness centralities  ${B}=\{\betw(w), w\in V\}$.
The algorithm samples from the domain $\domain\leftarrow\{(u,v)\in V\times V, u\neq v\}$.
%which control the quality of the
%approximation. At a high level, the algorithm works as follows. First, it
%computes an initial sample size $S_1$ (details given later), then samples pairs
%of different nodes independently and uniformly at random with replacement from
%the domain $\domain=\{(u,v)\in V\times V ~:~ u\neq v\}$, and for any $w\in V$.
%For each sampled pair $(u,v)$, the algorithm runs an $s-t$ SP computation
%from $u$ to $v$, updates the \BC estimation of every node on any SP
%from $u$ to $v$, and performs some bookkeeping needed for checking the stopping
%condtion. After $S_1$ pairs have been sampled, \staticalgo verifies the stopping
%condition, whose correctness is rooted on~\Cref{thm:supdevbound,thm:radeboundw}
%and which depends on the quality-controlling parameters $\varepsilon$ and
%$\delta$. If the stopping condition is satisfied, \staticalgo outputs the \BC
%estimation for all nodes, otherwise it computes a new sample size $S_2$ and
%iterates. At the next iteration, $S_1-S_2$ additional pairs of nodes are
%sampled, and the algorithm uses information obtained from the $S_2$ sampled
%pairs to compute the \BC estimations and checking the stopping condition, and so
%on.

%\paragraph*{Detailed Description}
%We now present \staticalgo in details.
For each node $w\in V$, let $f_w : \domain \to \mathbb{R}^+$ be the function
\begin{equation}\label{eq:functionf}
	f_w(u,v)=\frac{\sigma_{uv}(w)}{\sigma_{uv}},
\end{equation}
i.e., $f_w(u,v)$ is the fraction of shortest paths (SPs) from $u$ to $v$ that go through $w$. Let
$\family$ be the set of these functions. Given this definition, we have that
\begin{align*}
	\mean_\domain(f_w)&=\frac{1}{|\domain|}\sum_{(u,v)\in\domain}f_w(u,v)\\
	&=\frac{1}{|V|(|V|-1)}\sum_{\substack{(u,v)\in V\times V\\u\neq
		v}}\frac{\sigma_{uv}(w)}{\sigma_{uv}}\\
	&=\betw(v)\enspace.
\end{align*}
Let now $\Sam=\{(u_i,v_i), 1\le i \le \ell\}$ be a collection of $\ell$ pairs
$(u,v)$ from $\domain$. For the sake of clarity, we define
\[
		\tbetw(w)=\mean_\Sam(f_w)=\frac{1}{\ell}\sum_{i=1}^\ell
		f_w((u_i,v_i))\enspace.
\]
For each $w\in V$ consider the vector
\[
		\mathbf{v}_w=(f_w(u_1,v_1),\dotsc,f_w(u_\ell,v_\ell))\enspace.
\]
It is easy to see that $\tbetw(w)=\|\mathbf{v}_w\|_1/\ell$. Let now
$\mathcal{V}_\Sam$ be the \emph{set} of these vectors:
\[
	\mathcal{V}_\Sam=\{\mathbf{v}_w, w\in V\}\enspace.
\]
If we have complete knowledge of this set of vectors, then we can compute the
quantity
\[
		\omega^*=\min_{s\in\mathbb{R}^+}\frac{1}{s}\ln\sum_{\mathbf{v}\in\mathcal{V}_\Sam}\mathrm{exp}\left(s^2\|\mathbf{v}\|^2/(2\ell^2)\right),
\]
then use $\omega^*$ in~\eqref{eq:alpha} in place of $\rade(\family,\Sam)$ to
obtain $\alpha$, and combine~\eqref{eq:supdevboundrefined},
\eqref{eq:functionw}, and \eqref{eq:functionwbound} to obtain
\begin{equation}\label{eq:deltai}
	\Delta_\Sam= \frac{\omega^*}{1-\alpha} +
	\frac{\ln\frac{2}{\delta}}{2\ell\alpha(1-\alpha)}+\sqrt{\frac{\ln\frac{2}{\delta}}{2\ell}},
\end{equation}
and finally check whether $\Delta_\Sam\le\varepsilon$. This is \staticalgo's
stopping condition. When it holds, we can just return the collection
$\widetilde{B}=\{\tbetw(w)=\|\mathbf{v}_w\|_1/\ell, w\in V\}$ since, from the
definition of $\Delta_\Sam$ and~\Cref{thm:supdevboundrefined,thm:radeboundw}, we
have that $\widetilde{B}$ is an $(\varepsilon,\delta)$-approximation to the
exact betweenness values.

\staticalgo works as follows. Suppose for now that we fix a priori a
monotonically increasing sequence $(S_i)_{i>0}$ of sample sizes (we show in
later paragraph how to compute the sample schedule adaptively on the fly). The
algorithm builds a collection $\Sam$ by sampling pairs $(u,v)$ independently
and uniformly at random from $\domain$, until it reaches size $S_1$. After each
pair of nodes has been sampled, \staticalgo performs an $s-t$ SP computation
from $u$ to $v$ and then backtracks from $v$ to $u$ along the SPs just computed,
to keeps track of the set $\mathcal{V}_{\Sam}$ of vectors (details given below).
For clarity of presentation, let $\Sam_1$ denote $\Sam$ when it has size exactly
$S_1$, and analogously for $\Sam_i$ and $S_i$, $i>1$. Once $\Sam_1$ has been
built, \staticalgo computes $\Delta_{\Sam_i}$ and checks whether it is at most
$\varepsilon$. If so, then it returns $\widetilde{B}$.  Otherwise, \staticalgo
iterates and continues adding samples from $\domain$ to $\Sam$ until it has size
$S_2$, and so on until $\Delta_{\Sam_i}\le\varepsilon$ holds. The pseudocode for
\staticalgo is presented in Alg.~\ref{alg:static}, including the steps to
update $\mathcal{V}_\Sam$ and to adaptively choose the sample schedule, as
described in the following paragraphs. We now prove the correctness of the
algorithm.

\begin{algorithm}[h!t]
	\ifappendix
	\else
	\small
	\fi
	\DontPrintSemicolon
	\SetKwInOut{Input}{input}
	\SetKwInOut{Output}{output}
	\SetKwFunction{GetSample}{uniform\_random\_sample}
	\SetKwFunction{ModifiedSP}{compute\_SPs}
	\SetKwComment{tcp}{//}{}
	\Input{Graph $G=(V,E)$, accuracy parameter $\varepsilon\in(0,1)$, confidence parameter $\delta\in(0,1)$}
	\Output{Set $\widetilde{B}$ of \BC approximations for all nodes in $V$}
	$\domain\leftarrow\{(u,v)\in V\times V, u\neq v\}$\;
	$S_0\leftarrow 0$, $S_1\leftarrow
	\frac{(1+8\varepsilon+\sqrt{1+16\varepsilon})\ln(2/\delta)}{4\varepsilon^2}$
	%9\ln(1/\delta)/\varepsilon^2$
	\label{algline:initialsize}\;
	$\bm0=(0)$\label{algline:bmo}\;
	$\mathcal{V}=\{\bm 0\}$\;
	\lForEach{$w\in V$}{ $M[w]=\bm 0$}
	$\mathsf{c}_{\bm 0} \leftarrow |V|$\;
	$j\leftarrow 1$\;
	\While{True} {\label{algline:mainloop}
		\For{$j\leftarrow 1$ \KwTo $S_i-S_{i-1}$}{\label{algline:internalloop}
			$(u,v)\leftarrow$ \GetSample{$\domain$}\label{algline:getsample}\;
			\tcp{Truncated SP computation}
			\ModifiedSP{$u,v$}\label{algline:spcomp}\;
			\If{reached $v$} {
				\lForEach{$z\in\mathsf{P}_u[v]$}{ $\sigma_{zv}\leftarrow 1$}
				\ForEach{node $w$ on a SP from $u$ to $v$, in reverse order by
					$\mathsf{d}(u,w)$} {\label{algline:backtrack}
					$\sigma_{uv}(w)\leftarrow\sigma_{uw}\sigma_{wv}$\;
					$\mathbf{v}\leftarrow M[w]$\label{algline:beginupdate}\;
					$\mathbf{v}'\leftarrow(\underbrace{(j_1, g_1), (j_2,
							g_2), \dotsc }_\mathbf{v}, (j,\sigma_{uv}(w)))$\;
					\uIf{$\mathbf{v}'\not\in\mathcal{V}$}{
						$c_{\mathbf{v}'}\leftarrow 1$\;
						$\mathcal{V}\leftarrow \mathcal{V}\cup\{\mathbf{v}'\}$\;
					}
					\lElse{ $\mathsf{c}_{\mathbf{v}'}\leftarrow \mathsf{c}_{\mathbf{v}'} +1$}
					$M[w]\leftarrow \mathbf{v}'$\label{algline:endupdate}\;
					\lIf{$\mathsf{c}_\mathbf{v} > 1$} {
						$\mathsf{c}_{\mathbf{v}}\leftarrow \mathsf{c}_{\mathbf{v}} -1$
					}
					\lElse {$\mathcal{V}\leftarrow\mathcal{V}\setminus\{\mathbf{v}\}$
					\label{algline:remove}}
					\lForEach{$z\in\mathsf{P}_u[w]$}{
						$\sigma_{zv}\leftarrow \sigma_{zv}+\sigma_{wv}$\label{algline:recurscomp}
					}
				}
			}
			\label{algline:internalloopend}
		}
		$\omega^*_i\leftarrow\min_{s\in\mathbb{R}^+}\frac{1}{s}\ln\sum_{\mathbf{v}\in\mathcal{V}_\Sam}\mathrm{exp}\left(s^2\|\mathbf{v}\|^2/(2S_{i}^2)\right)$\;
		$\alpha_i\leftarrow\frac{\ln\frac{2}{\delta}}{\ln\frac{2}{\delta}+\sqrt{\left(2S_{i}\omega^*_i+\ln\frac{2}{\delta}\right)\ln\frac{2}{\delta}}}$\;
			$\Delta_{\Sam_i}\leftarrow\frac{\omega^*_i}{1-\alpha_i} +
			\frac{\ln\frac{2}{\delta}}{2S_{i}\alpha_i(1-\alpha_i)}+\sqrt{\frac{\ln\frac{2}{\delta}}{2S_{i}}}$\label{algline:getdelta}\;
		\lIf{$\Delta_{\Sam_i}\le\varepsilon$}{
			\textbf{break}
		}
		\Else{
			$S_{i+1}\leftarrow $ \texttt{nextSampleSize()}
			%\left(\Delta_{\Sam_i} / \varepsilon\right)^2 S_i$
			\label{algline:nextsize}\;
			$i\leftarrow i+1$\;
		}
		$j\leftarrow j+1$\;
	}
	\Return{$\widetilde{B}\leftarrow\{\tbetw(w)\leftarrow\|M[w]\|_1/S_i, w\in V\}$}\;
	\caption{\staticalgo: absolute error approximation of \BC on static graphs}
	\label{alg:static}
\end{algorithm}

\begin{theorem}[correctness]\label{thm:correctness}
	The collection $\widetilde{B}$ returned by \staticalgo is a
	$(\varepsilon,\delta)$-approximation to the collection of exact \BC values.
\end{theorem}

\begin{proof}
	The claim follows from the definitions of $\Sam$, $\mathcal{V}_\Sam$,
	$\family$, $f_w$ for $w\in V$, $\tbetw(w)$, $\Delta_{\Sam_i}$, and
	from~\Cref{thm:supdevboundrefined,thm:radeboundw}.
\end{proof}

\paragraph*{Computing and maintaining the set $\mathcal{V}_\Sam$}
We now discuss in details how \staticalgo can efficiently maintain the set
$\mathcal{V}_\Sam$ of vectors, which is used to compute the value
$\Delta_{\Sam}$ and the values $\tbetw(w)=\|\mathbf{v}_w\|_1/|\Sam|$ in $\widetilde{B}$. In addition
to $\mathcal{V}_\Sam$, \staticalgo also maintains a map $M$ from $V$ to
$\mathcal{V}_\Sam$ (i.e., $M[w]$ is a vector $\mathbf{v}_w\in\mathcal{V}_\Sam$),
and a counter $\mathsf{c}_\mathbf{v}$ for each $\mathbf{v}\in\mathcal{V}_\Sam$,
denoting how many nodes $w\in V$ have $M[w]=\mathbf{v}$.

At the beginning of the execution of the algorithm, we have $\Sam=\emptyset$ and
also $\mathcal{V}_\Sam=\emptyset$. Nevertheless, \staticalgo initializes
$\mathcal{V}_\Sam$ to contain one special empty vector $\bm 0$, with no
components, and $M$ so that $M[w]=\bm 0$ for all $w\in V$, and $\mathsf{c}_{\bm
0}=|V|$ (lines~\ref{algline:bmo} and following~in Alg:~\ref{alg:static}).

After having sampled a pair $(u,v)$ from $\domain$, \staticalgo updates
$\mathcal{V}_\Sam$, $M$ and the counters as follows. First, it performs
(line~\ref{algline:spcomp}) a $s-t$
SP computation from $u$ to $v$ using any SP algorithms (e.g., BFS or Dijkstra)
modified, as discussed by~\citet[Lemma 3]{Brandes01}, to keep track, for each
node $w$ encountered during the computation, of the SP distance
$\mathsf{d}(u,w)$ from $u$ to $w$, of the number $\sigma_{uw}$ of SPs from $u$
to $w$, and of the set $\mathsf{P}_u(w)$ of (immediate) predecessors of $w$
along the SPs from $u$.\footnote{Storing the set of immediate predecessors is
not necessary. By not storing it, we can reduce the space complexity from
$O(|E|)$ to $O(|V|)$, at the expense of some additional computation at runtime.}
Once $v$ has been reached (and only if it has been reached), the algorithm
starts backtracking from $v$ towards $u$ along the SPs it just computed
(line~\ref{algline:backtrack}). During
this backtracking, the algorithm visits the nodes along the SPs in inverse
order of SP distance from $u$. For each visited node $w$ different from $u$ and
$v$, it computes the value $f_w(u,v)=\sigma_{uv}(w)$ of
SPs from $u$ to $v$ that go through $w$, which is obtained as
\[
\sigma_{uv}(w)=\sigma_{uw}\times\sum_{z~:~w\in\mathsf{P}_u(z)}\sigma_{zv}
\]
where the value $\sigma_{uw}$ is obtained during the $s-t$ SP computation, and
the values $\sigma_{zw}$ are computed recursively during the
backtracking (line~\ref{algline:recurscomp})~\citep{Brandes01}. After computing
$\sigma_{uv}(w)$, the algorithm
takes the vector $\mathbf{v}\in\mathcal{V}_{\Sam}$ such that $M[w]=\mathbf{v}$
and creates a new vector $\mathbf{v}'$ by appending $\sigma_{uv}(w)$ to the end
of $\mathbf{v}$.\footnote{\staticalgo uses a sparse representation for the
vectors $\mathbf{v}\in\mathcal{V}_\Sam$, storing only the non-zero components
of each $\mathbf{v}$ as pairs $(i,g)$, where $i$ is the component index and $g$
is the value of that component.} Then it adds $\mathbf{v}'$ to the \emph{set}
$\mathcal{V}_{\Sam}$, updates $M[w]$ to $\mathbf{v}'$, and increments the
counter $\mathsf{c}_{\mathbf{v}'}$ by one (lines~\ref{algline:beginupdate}
to~\ref{algline:endupdate}). Finally, the algorithm decrements the
counter $\mathsf{c}_{\mathbf{v}}$ by one, and if it becomes equal to zero,
\staticalgo removes $\mathbf{v}$ from $\mathcal{V}_\Sam$
(line~\ref{algline:remove}). At this point, the
algorithm moves to analyzing another node $w'$ with distance from $u$ less or
equal to the distance of $w$ from $u$. It is easy to see that when the
backtracking reaches $u$, the set $\mathcal{V}_\Sam$, the map $M$, and the
counters, have been correctly updated.

We remark that to compute $\Delta_{\Sam_i}$ and $\widetilde{B}$
and to keep the map $M$ up to date, we do not actually need to store the
vectors in $\mathcal{V}_\Sam$ (even in sparse form), but it is sufficient to
maintain their $\ell_1$- and Euclidean norms, which require much less space.

\subsubsection{Computing the sample schedule}\label{sec:samplesched}
We now discuss how to compute the initial sample size $S_1$ at the beginning
of \staticalgo (line~\ref{algline:initialsize} of Alg.~\ref{alg:static}) and
the sample size $S_{i+1}$ at the end of iteration $i$ of the main
loop (line~\ref{algline:nextsize}). We remark that any sample schedule
$(S_i)_{i>0}$ can be used, and our method is an heuristic that nevertheless
exploits all available information at the end of each iteration to the most
possible extent, with the goal of increasing the chances that the stopping
condition is satisfied at the next iteration.

As initial sample size $S_1$ we choose
\begin{equation}\label{eq:firstsamplesize}
		S_1\ge
		\frac{(1+8\varepsilon+\sqrt{1+16\varepsilon})\ln(2/\delta)}{4\varepsilon^2}\enspace.
\end{equation}
To understand the intuition behind this choice,
recall~\eqref{eq:supdevboundrefined}, and consider that, at the beginning of the
algorithm, we obviously have no information about $\rade(\family,\Sam_1)$, except
that it is \emph{non-negative}. Consequently we also can not compute $\alpha$ as
in~\eqref{eq:alpha}, but we can easily see that $\alpha\in[0,1/2]$. From the
fact that $\rade(\family,\Sam)\ge 0$, we have that, for the
r.h.s.~of~\eqref{eq:supdevboundrefined} to be at most $\varepsilon$ (i.e., for
the stopping condition to be satisfied after the first iteration of the
algorithm), it is necessary that
\[
	\frac{\ln\frac{2}{\delta}}{2S_1\alpha(1-\alpha)}+\sqrt{\frac{\ln\frac{2}{\delta}}{2S_1}}
	\le\varepsilon\enspace.
\]
Then, using the fact that the above expression decreases as $\alpha$
increases, we use $\alpha=1/2$, i.e., its maximum attainable value, to obtain
the following inequality, where $S_1$ acts as the unknown:
\[
	\frac{2\ln(2/\delta)}{S_1}+\sqrt{\frac{\ln(2/\delta)}{2S_1}}\le\varepsilon\enspace.
\]
Solving for $S_1$ under the constraint of $S_1\ge 1$, $\delta\in(0,1)$,
$\varepsilon\in(0,1)$ gives the unique solution in~\eqref{eq:firstsamplesize}.

Computing the next sample size $S_{i+1}$ at the end of iteration $i$ (in the
pseudocode in Alg.~\ref{alg:static}, this is done by calling
\texttt{nextSampleSize()} on line~\ref{algline:nextsize}) is slightly
more involved. The intuition is to assume that $\omega^*_i$, which is an
upper bound to $\rade(\family,\Sam_i)$, is also an upper bound to
$\rade(\family,\Sam_{i+1})$, whatever $\Sam_{i+1}$ will be, and whatever size it
may have. At this point, we can ask what is the minimum size
$S_{i+1}=|\Sam_{i+1}|$ for which $\Delta_{\Sam_{i+1}}$ would be at most
$\varepsilon$, under the assumption that
$\rade(\family,\Sam_{i+1})\le\omega^*_i$. More formally, we want to solve the
inequality
\begin{align}\label{eq:nextsampleinequal}
	%\sqrt{\frac{\ln\frac{2}{\delta}}{2S_{i+1}}}&+\frac{\omega^*_i}{1-\frac{\ln\frac{2}{\delta}}{\ln\frac{2}{\delta}+\sqrt{(2S_{i+1}\omega^*_i+\ln\frac{2}{\delta})\ln\frac{2}{\delta}}}}\\
	%&+\frac{\ln\frac{2}{\delta}}{2S_{i+1}\frac{\ln\frac{2}{\delta}}{\ln\frac{2}{\delta}+\sqrt{(2S_{i+1}\omega^*_i+\ln\frac{2}{\delta})\ln\frac{2}{\delta}}}\left(1-\frac{\ln\frac{2}{\delta}}{\ln\frac{2}{\delta}+\sqrt{(2S_{i+1}\omega^*_i+\ln\frac{2}{\delta})\ln\frac{2}{\delta}}}\right)}<\varepsilon
	&\left(1+\frac{\ln\frac{2}{\delta}}{\sqrt{(2S_{i+1}\omega^*_i+\ln\frac{2}{\delta})\ln\frac{2}{\delta}}}
\right)\nonumber\\
&\times\left(\omega^*_i+\frac{\ln\frac{2}{\delta}+\sqrt{(2S_{i+1}\omega^*_i+\ln\frac{2}{\delta})\ln\frac{2}{\delta}}}{2S_{i+1}} \right)
	+ \sqrt{\frac{\ln\frac{2}{\delta}}{2S_{i+1}}}\le\varepsilon
\end{align}
where $S_{i+1}$ acts as the unknown. The l.h.s.~of this inequality is obtained
by plugging~\eqref{eq:alpha} into~\eqref{eq:supdevboundrefined} and using
$\omega^*_i$ in place of $\rade(\family,\Sam)$, $S_{i+1}$ in place of
$\ell$, and slightly reorganize the terms for readability. Finding the solution
to the above inequality requires computing the roots of the cubic equation (in
$x$)
\begin{align}\label{eq:cubic}
	-8&\left(\ln\frac{2}{\delta}\right)^3+\left(\ln\frac{2}{\delta}\right)^2(-16\omega^*_i+(1+4\varepsilon)^2)x\nonumber\\
	&-4\left(\ln\frac{2}{\delta}\right)(\omega^*_i-\varepsilon)^2(1+4\varepsilon)x^2
	+ 4(b-f)^4 x^3=0\enspace.
\end{align}
One can verify that the roots of this equation are all reals. The roots are
presented in~\Cref{tab:roots}. The solution to
inequality~\eqref{eq:nextsampleinequal} is that $S_{i+1}$ should be larger than
one of these roots, but which of the roots it should be larger than depends on
the values of $\omega^*_i$, $\delta$, and $\varepsilon$. In
practice, we compute each of the roots and then choose the smallest positive one
such that, when $S_{i+1}$ equals to this root, then~\eqref{eq:nextsampleinequal}
is satisfied.

\begin{table*}[t]
	\ifappendix
	\else
	\small
	\fi
	\centering
	\begin{tabular}{lc}
		\toprule
		\multicolumn{2}{l}{Let  $\left\{\begin{array}{l}z=48\omega^*_i+(1+4\varepsilon)^2\\
			w=-1-12\varepsilon+8(27(\omega^*_i)^2+(21-8\varepsilon)\varepsilon^2+18b(1+f)\\
	y=12\sqrt{3}|-1+2\omega^*_i+2\varepsilon|\sqrt{-(27(\omega^*_i)^2
-\varepsilon^2(1+16\varepsilon)-\omega^*_i(1+18\varepsilon))}\\
\theta=\arg(-w+jy)/3 \mbox{ where $j$ is the imaginary unity and $\arg(\ell)$ is
	the argument of the complex number $\ell$}\end{array}\right.$}\\
		\midrule
		Root 1 &
		$\frac{1}{3}(\ln\frac{2}{\delta})((1+4\varepsilon) -
		\sqrt{z}\cos\theta)(\omega^*_i-\varepsilon)^{-2} $\\
		Root 2 & $\frac{1}{6}(\ln\frac{2}{\delta})(2(1+4\varepsilon) +
\sqrt{z}(\cos\theta+\sqrt{3}\sin\theta))(\omega^*_i-\varepsilon)^{-2}$ \\
Root 3 & $\frac{1}{6}(\ln\frac{2}{\delta})(2(1+4\varepsilon) +
\sqrt{z}(\cos\theta-\sqrt{3}\sin\theta))(\omega^*_i-\varepsilon)^{-2}$ \\
		\bottomrule
	\end{tabular}
	\caption{Roots of the cubic equation~\eqref{eq:cubic} for the computation of
	the next sample size.}
	\label{tab:roots}
\end{table*}

The assumption $\rade(\family,\Sam_{i+1})\le\omega^*_i$, which is not guaranteed
to be true, is what makes our procedure for selecting the next sample size an
\emph{heuristics}.  Nevertheless, Using information available at the current
iteration to compute the sample size for the next iteration is more sensible
than having a fixed sample schedule, as it tunes the growth of the sample size
to the quality of the current sample. Moreover, it removes from the user the
burden of choosing a sample schedule, effectively eliminating one parameter of
the algorithm.
%To compute the next sample size $S_{i+1}$ at the end of iteration $i$, we use
%the value $\Delta_{i}$ computed during iteration $i$, and set
%\begin{equation}\label{eq:nextsample}
%	S_{i+1}=\frac{\Delta_{i}}{\varepsilon}^2 S_i.
%\end{equation}
%The intuition behind this choice is that if everything stays the same
%and only the sample size changes, then $\Delta_{i+1}$ would be equal to
%$\varepsilon$, hence the algorithm would stop.

\subsection{Relative-error Top-k Approximation}\label{sec:topk}
In practical applications it is usually necessary (and sufficient) to identify
the vertices with highest \BC, as they act, in some sense, as the ``primary
information gateways'' of the network. In this section we present a variant
\staticalgotopk of \staticalgo to compute a high-quality approximation of the
set $\TOP(k,G)$ of the top-$k$ vertices with highest \BC in a graph $G$. The
approximation $\tbetw(w)$ returned by \staticalgotopk for a node $w$ is within a
\emph{multiplicative} factor $\varepsilon$ from its exact value $\betw(w)$,
rather than an additive factor $\varepsilon$ as in \staticalgo. This higher
accuracy has a cost in terms of the number of samples needed to compute the
approximations.

Formally, assume to order the nodes in the graph in decreasing order by \BC,
ties broken arbitrarily, and let $b_k$ be the \BC of the $k$-th node in this
ordering. Then the set $\TOP(k,G)$ is defined as the set of nodes with \BC at
least $b_k$, and can contain more than $k$ nodes:
\[
	\TOP(k,G)=\{(w,\betw(w) ~:~ v\in V \mbox{ and } \betw(w)\ge b_k\}\enspace.
\]

The algorithm \staticalgotopk follows the same approach
as the algorithm for the same task by~\citet[Sect.~5.2]{RiondatoK15} and works
in two phases. Let $\delta_1$ and $\delta_2$ be such that
$(1-\delta_1)(1-\delta_2)\ge(1-\delta)$. In the first phase, we run \staticalgo
with parameters $\varepsilon$ and $\delta_1$. Let $\ell'$ be the $k$-th highest
value $\tbetw(w)$ returned by \staticalgo, ties broken arbitrarily, and let
$\tilde{b}'=\ell'-\varepsilon$.

In the second phase, we use a variant \staticalgorel of \staticalgo with a
modified stopping condition based on relative-error versions
of~\Cref{thm:supdevbound,thm:radeboundw}
\ifappendix
 (\Cref{thm:raderel,thm:radeboundwr} from~\Cref{app:raderel})
\else
 (Thms.~11 and~12 from Appendix~D of the extended
 online version~\citep{RiondatoU16ext})
\fi
, which take $\varepsilon$, $\delta_2$,
and $\lambda=\tilde{b}'$ as parameters. The parameter $\lambda$ plays a role in
the stopping condition. Indeed, \staticalgorel is the same as \staticalgo, with
the only crucial difference in the definition of the quantity $\Delta_{\Sam_i}$,
which is now:
\begin{equation}\label{eq:deltairel}
	\Delta_{\Sam_i} =
	2\min_{s\in\mathbb{R}^+}\frac{1}{s}\ln\displaystyle\sum_{\mathbf{v}
	\in\mathcal{V}}\mathrm{exp}\left(\frac{s^2\|\mathbf{v}\|^2}{\lambda2S_i^2}\right)
		+ \frac{3}{\lambda}\sqrt{\frac{\ln(2/\delta)}{2S_i}}\enspace.
\end{equation}

\begin{theorem}\label{thm:staticalgorel}
	Let
	\[
		\widetilde{B}=\{\tbetw(w), w\in V\}
	\]
	bet the output of \staticalgorel. Then $\widetilde{B}$ is such that
	\[
		\Pr\left(\exists w\in V ~:~
		\frac{|\tbetw(v)-\betw(v)|}{\max\{\lambda,\betw(v)\}}>\varepsilon\right)<\delta\enspace.
	\]
\end{theorem}

The proof follows the same steps as the proof for~\Cref{thm:correctness}, using
the above definition of $\Delta_{\Sam_i}$ and
applying
\ifappendix
\Cref{thm:raderel,thm:radeboundwr} from~\Cref{app:raderel} %
\else
Thms.~11 and~12 from Appendix~D of the extended online version~\citep{RiondatoU16ext} %
\fi
 instead of~\Cref{thm:supdevboundrefined,thm:radeboundw}.

Let $\ell''$ be the $k$-th highest value $\tbetw(w)$ returned by \staticalgorel
and let $\tilde{b}''=\ell''/(1+\varepsilon)$. \staticalgotopk then returns the
set
\[
	\widetilde{\TOP}(k,G)=\{(w,\tbetw(w)) ~:~ w\in V \mbox{ and } \tbetw(w)\ge
	\tilde{b}''\}\enspace.
\]
We have the following result showing the properties of the collection
$\widetilde{\TOP}(k,G)$.

\begin{theorem}\label{thm:topk}
	With probability at least $1-\delta$, the set $\widetilde{\TOP}(k,G)$ is
	such that:
	\begin{enumerate}
		\item for any pair $(v,\betw(v))\in\TOP(k,G)$, there is one pair
			$(v,\tbetw(v))\in\widetilde{\TOP}(k,G)$ (i.e., we return a superset
			of the top-$k$ nodes with highest betweenness) and this pair is such
			that $|\tbetw(w)-\betw(w)|\le\varepsilon\betw(w)$;
		\item for any pair $(w,\tbetw(w))\in\widetilde{\TOP}(k,G)$ such that
			$(w,\betw(w))\not\in\TOP(k,G)$ (i.e., any false positive) we have
			that $\tbetw(w)\le(1+\varepsilon)b_k$ (i.e., the false positives, if
			any, are among the nodes returned by \staticalgotopk with lower \BC estimation).
	\end{enumerate}
\end{theorem}

The proof and the pseudocode for \staticalgotopk can
be found in
\ifappendix
\Cref{app:topk}. %
\else
Appendix~A of the extended online version~\citep{RiondatoU16ext}. %
\fi

\subsection{Special Cases}\label{sec:unique}
In this section we consider some special restricted settings that make computing
an high-quality approximation of the \BC of all nodes easier. One example of
such restricted settings is when the graph is \emph{undirected} and every pair
of distinct nodes is either connected with a \emph{single} SP or there is no
path between the nodes. This is the case for many road networks, where the
unique SP condition is often enforced~\citep{GeisbergerSS08}. \citet[Lemma
2]{RiondatoK15} showed that, in this case, the number of samples needed to
compute a high-quality approximation of the \BC of all nodes is
\emph{independent} on any property of the graph, and only depends on the quality
controlling parameters $\varepsilon$ and $\delta$. The algorithm
by~\citet{RiondatoK15} works differently from \staticalgo, as it samples one SP
at a time and only updates the \BC estimation of nodes along this path, rather
than sampling a pair of nodes and updating the estimation of all nodes on any
SPs between the sampled nodes. Nevertheless, as shown in the following theorem,
we can actually even generalize the result by~\citet{RiondatoK15}, as
shown in~\Cref{thm:unique}. The statement and the proof of this theorem use
pseudodimension~\citep{Pollard84}, an extension of the Vapnik-Chervonenkis (VC)
dimension to real-valued functions. Details about pseudodimension and the proof
of~\Cref{thm:unique} can be found in
\ifappendix
\Cref{app:unique}. %
\else
Appendix~B of the extended online version~\citep{RiondatoU16ext}. %
\fi
\Cref{corol:unique} shows how to modify \staticalgo to take~\Cref{thm:unique}
into account.

\begin{theorem}\label{thm:unique}
	Let $G=(V,E)$ be a graph such that it is possible to partition the set
	$\domain=\{(u,v)\in V\times V, u\neq v\}$ in two classes: a class $A=\{(u^*,
	v^*)\}$ containing a single pair of different nodes $(u^*,v^*)$ such that
	$\sigma_{u^*v^*}\le 2$ (i.e., connected by either at most two SPs or not
	connected), and a class $B=\domain\setminus A$ of pairs $(u,v)$ of nodes
	with $\sigma_{uv}\le 1$ (i.e., either connected by a single SP or not
	connected). Then the pseudodimension of the family of functions
	\[
		\{f_w ~:~ \domain \to [0,1], w\in V\},
	\]
	where $f_w$ is defined as in~\eqref{eq:functionf}, is at most $3$.
\end{theorem}

\begin{corollary}\label{corol:unique}
	Assume to modify \staticalgo with the additional stopping condition
	instructing to return the set $\tilde{B}=\{\tbetw(w), w\in V\}$ after a
	total of
	\[
		r=\frac{c}{\varepsilon^2}\left(3+\ln\frac{1}{\delta}\right)
	\]
	pairs of nodes have been sampled from $\domain$. %over the course of the algorithm. (i.e., after $r$ total iterations of the internal loop).
	The set $\tilde{B}$ is s.t.%still such that
	\[
		\Pr(\exists w\in V ~:~ |\tbetw(w)-\betw(w)|>\varepsilon)<\delta\enspace.
	\]
\end{corollary}

The bound in~\Cref{thm:unique} is strict, i.e., there exists a graph for which
the pseudodimension is exactly $3$~\citep[Lemma 4]{RiondatoK15}. Moreover, as
soon as we relax the requirement in~\Cref{thm:unique} and allow two pairs of
nodes to be connected by two SPs, there are graphs with pseudodimension $4$
\ifappendix
(\Cref{lem:uniquetight} in \Cref{app:unique}). %
\else
(Lemma~4 in Appendix~B of the extended online
version~\citep{RiondatoU16ext}).
\fi

For the case of \emph{directed} networks, it is currently an open
question whether a high-quality (i.e., within $\varepsilon$) approximation of
the \BC of all nodes can be computed from a sample whose size is independent
of properties of the graph, but it is known that, even if possible, the
constant would not be the same as for the undirected
case~\citep[Sect.~4.1]{RiondatoK15}.

We conjecture that, given some information on how many pair of nodes are
connected by $x$ shortest paths, for $x\ge 0$, it should be possible to derive a
strict bound to the pseudodimension associated to the graph.

\subsection{Improved Estimators}\label{sec:imprest}
\citet{GeisbergerSS08} present an improved estimator for \BC using random
sampling. Their experimental results show that the quality of the
approximation is significantly improved, but they do not present any theoretical
analysis. Their algorithm, which follows the work of~\citet{BrandesP07} differs
from ours as it samples vertices and performs a Single-Source-Shortest-Paths (SSSP) computation from each of
the sampled vertices. We can use an adaptation of their estimator in a variant
of our algorithm, and we can prove that this variant is still probabilistically
guaranteed to compute an $(\varepsilon,\delta)$-approximation of the \BC of all
nodes, therefore removing the main limitation of the original work, which
offered no quality guarantees. We now present this variant considering, for ease
of discussion, the special case of the linear scaling estimator
by~\citet{GeisbergerSS08}, this technique can be extended to the generic
parameterized estimators they present.

The intuition behind the improved estimator is to increase the estimation of the
\BC for a node $w$ proportionally to the ratio between the SP distance
$\mathsf{d}(u,w)$
from the first component $u$ of the pair $(u,v)$ to $w$ and the SP
distance $\mathsf{d}(u,v)$ from $u$ to $v$. Rather than sampling pairs
of nodes, the algorithm samples triples $(u,v,d)$, where $d$ is a \emph{direction},
(either $\leftarrow$ or $\rightarrow$), and updates the betweenness estimation
differently depending on $d$, as follows. Let
$\domain'=\domain\times\{\leftarrow,\rightarrow\}$ and for each $w\in
V$, define the function $g_w$ from $\domain'$ to $[0,1]$ as:
\[
	g_w(u,v,d)=\left\{\begin{array}{ll}
		\frac{\sigma_{uv}(w)}{\sigma_{uv}}\frac{\mathsf{d}(u,w)}{\mathsf{d}(u,v)} & \mbox{if } d=\rightarrow\\
			\frac{\sigma_{uv}(w)}{\sigma_{uv}}\left(1-\frac{\mathsf{d}(u,w)}{\mathsf{d}(u,v)}\right) & \mbox{if } d=\leftarrow
	\end{array}\right.
\]
Let $\Sam$ be a collection of $\ell$ elements of $\domain'$ sampled uniformly
and independently at random with replacement. Our estimation $\tbetw(w)$ of the
\BC of a node $w$ is
\[
	\tbetw(w)=\frac{2}{\ell}\sum_{(u,v,d)\in\Sam}g_w(u,v,d)=2\mean_\Sam(f_w)\enspace.
\]
The presence of the factor $2$ in the estimator calls for a single minor
adjustment in the definition of $\Delta_{\Sam_i}$ which, for this variant of
\staticalgo, becomes
\[
		\Delta_{\Sam_i} =
		\frac{\omega^*_i}{1-\alpha_i} +
			\frac{\ln\frac{2}{\delta}}{2S_{i}\alpha_i(1-\alpha_i)}+\sqrt{\frac{2\ln\frac{2}{\delta}}{S_{i}}}
		%2\min_{s\in\mathbb{R}^+}\frac{1}{s}\ln\sum_{\mathbf{v}\in\mathcal{V}_\Sam}\mathrm{exp}\left(s^2\|\mathbf{v}\|^2/2S_i^2\right) + 3\sqrt{4\frac{\ln(2/\delta)}{S_i}},
\]
i.e., w.r.t.~the original definition of $\Delta_{\Sam_i}$, there is an additional
factor $4$ inside the square root of the third term on the
r.h.s..

The output of this variant of \staticalgo is still a high-quality approximation
of the \BC of all nodes, i.e., \Cref{thm:correctness} still holds with this new
definition of $\Delta_{\Sam_i}$. This is due to the fact that the results on the
Rademacher averages presented in~\Cref{sec:radeprelims} can be extended to
families of functions whose co-domain is an interval $[a,b]$, rather than just
$[0,1]$~\citep{ShalevSBD14}.

\section{Dynamic Graph BC Approximation}\label{sec:dynamic}
In this section we  present an algorithm, named \dynamicalgo, that computes and
keeps up to date an high-quality approximation of the \BC of all nodes in a
\emph{fully dynamic graph}, i.e., in a graph where vertex and edges can be added
or removed over time. Our algorithm leverages on the recent work
by~\citet{HayashiAY15}, who introduced two fast data structures called the
Hypergraph Sketch and the Two-Ball Index: the Hypergraph Sketch stores the \BC
estimations for all nodes, while the Two-Ball Index is used to store the SP DAGs
and to understand which parts of the Hypergraph Sketch needs to be modified
after an update to the graph (i.e., an edge or vertex insertion or deletion).
\citet{HayashiAY15} show how to populated and update these data structures to
maintain an $(\varepsilon,\delta)$-approximation of the \BC of all nodes in a
fully dynamic graph. Using the novel data structures results in
orders-of-magnitude speedups w.r.t.~previous
contributions~\citep{BergaminiM15,BergaminiM15arXiv}. The algorithm
by~\citet{HayashiAY15} is based on a static random sampling approach which is
identical to the one described for \staticalgo, i.e., pairs of nodes are sampled
and the \BC estimation of the nodes along the SPs between the two nodes are
updated as necessary. Their analysis on the number of samples necessary to
obtain an $(\varepsilon,\delta)$-approximation of the \BC of all nodes uses the
union bound, resulting in a number of samples that depends on the logarithm of
the number of nodes in the graph, i.e., $O(\varepsilon^-2(\log(|V|/\delta)))$
pairs of nodes must be sampled.

\dynamicalgo builds and improves over the algorithm presented
by~\citet{HayashiAY15} as follows. Instead of using a static random sampling
approach with a fixed sample size, we use the progressive sampling approach and
the stopping condition that we use in \staticalgo to understand when we sampled
enough to first populate the Hypegraph Sketch and the Two-Ball Index. Then,
after each update to the graph, we perform the same operations as in the
algorithm by~\citet{HayashiAY15}, with the crucial addition, after these
operation have been performed, of keeping the set $\mathcal{V}_\Sam$ of vectors
and the map $M$ (already used in \staticalgo) up to date, and checking whether
the stopping condition is still satisfied. If it is not, additional pairs of
nodes are sampled and the Hypergraph Sketch and the Two-Ball Index are updated
with the estimations resulting from these additional samples. The sampling of
additional pairs continues until the stopping condition is satisfied,
potentially according to a sample schedule either automatic, or specified by the
user. As we show in~\Cref{sec:exper}, the overhead of additional checks of the
stopping condition is minimal. On the other hand, the use of the progressive
sampling scheme based on the Rademacher averages allows us to sample much fewer
pairs of nodes than in the static sampling case based on the union bound:
\citep{RiondatoK15} already showed that it is possible to sample much less than
$O(\log|V|)$ nodes, and, as we show in our experiments, our sample sizes are
even smaller than the ones by~\citep{RiondatoK15}. The saving in the number of
samples results in a huge speedup, as the running time of the algorithms are, in
a first approximation, linear in the number of samples, and in a reduction in
the amount of space required to store the data structures, as they now store
information about fewer SP DAGs.

%We have the following result about the correctness of \dynamicalgo.

\begin{theorem}\label{thm:dynamiccorrectness}
	The set $\widetilde{B}=\{\tbetw(w), w\in V\}$ returned by \dynamicalgo after
	each update has been processed is such that
	\[
		\Pr(\exists w\in V \mbox{ s.t. }
		|\tbetw(w)-\betw(w)|>\varepsilon)<\delta\enspace.
	\]
\end{theorem}

The proof follows from the correctness of the algorithm by~\citet{HayashiAY15}
and of \staticalgo (\Cref{thm:correctness}).

\section{Experimental Evaluation}\label{sec:exper}
In this section we presents the results of our experimental evaluation. We
measure and analyze the performances of \staticalgo in terms of its runtime and
sample size  and accuracy, and compared them with those of the exact algorithm
\textsf{BA}~\citep{Brandes01} and the approximation algorithm
\textsf{RK}~\citep{RiondatoK15}, which offers the same guarantees as \staticalgo
(computes an $(\varepsilon,\delta)$-approximation the \BC of all vertices).

\paragraph*{Implementation and Environment} We implement \staticalgo and
\dynamicalgo in C\raisebox{0.5ex}{\tiny\textbf{++}}, as an extension of the
NetworKit library~\citep{StaudtSM14}. The code is available
from~\url{http://matteo.rionda.to/software/ABRA-radebetw.tbz2}. We performed the
experiments on a machine with a AMD Phenom\textsuperscript{TM} II X4
955 processor and 16GB of RAM, running FreeBSD 11.

\paragraph*{Datasets and Parameters} We use graphs of various nature
(communication, citations, P2P, and social networks) from the SNAP
repository~\citep{LeskovecK14}. The characteristics of the graphs are reported
in the leftmost column of \Cref{tab:bigtable}.

In our experiments we varied $\varepsilon$ in the range $[0.005,0.3]$, and we also
evaluate a number of different sampling schedules (see
\Cref{sec:expersamplesched}). In all the results we report, $\delta$ is fixed to
$0.1$. We experimented with different values for this parameter, and, as
expected, it has a very limited impact on the nature of the results, given the
logarithmic dependence of the sample size on $\delta$. We performed five runs
for each combination of parameters. The variance between the different runs was
essentially insignificant, so we report, unless otherwise specified, the results
for a random run. %an average among the different runs.

\begin{table*}[ht]
	\scriptsize
	\centering
	\begin{tabular}{ccrrccccrcrcc}
		\toprule
		& & & \multicolumn{2}{c}{\shortstack{Speedup\\w.r.t.}} &
		\multicolumn{3}{c}{\shortstack[l]{Runtime\\Breakdown (\%)}} & & &
		\multicolumn{3}{c}{Absolute Error ($\times 10^5$)} \\
		\cmidrule(l{2pt}r{2pt}){4-5} \cmidrule(l{2pt}r{2pt}){6-8}
		\cmidrule(l{2pt}r{2pt}){11-13}

		Graph & $\varepsilon$ & \shortstack{Runtime\\(sec.)} & \textsf{BA} & \textsf{RK} & Sampling &
		\shortstack{Stop\\Cond.} & Other & \shortstack{Sample\\Size} &
		\shortstack{Reduction\\w.r.t.\\\textsf{RK}} & max & avg & stddev \\
		\midrule
	\multirow{6}{*}{\shortstack{Soc-Epinions1\\Directed\\$|V|=75,879$\\$|E|=508,837$}} & 0.005 & 483.06 & 1.36 & 2.90 & 99.983 & 0.014 & 0.002 & 110,705 & 2.64 & 70.84 & 0.35 & 1.14\\
 & 0.010 & 124.60 & 5.28 & 3.31 & 99.956 & 0.035 & 0.009 & 28,601 & 2.55 & 129.60 & 0.69 & 2.22\\
 & 0.015 & 57.16 & 11.50 & 4.04 & 99.927 & 0.054 & 0.018 & 13,114 & 2.47 & 198.90 & 0.97 & 3.17\\
 & 0.020 & 32.90 & 19.98 & 5.07 & 99.895 & 0.074 & 0.031 & 7,614 & 2.40 & 303.86 & 1.22 & 4.31\\
 & 0.025 & 21.88 & 30.05 & 6.27 & 99.862 & 0.092 & 0.046 & 5,034 & 2.32 & 223.63 & 1.41 & 5.24\\
 & 0.030 & 16.05 & 40.95 & 7.52 & 99.827 & 0.111 & 0.062 & 3,668 & 2.21 & 382.24 & 1.58 & 6.37\\
\midrule
\multirow{6}{*}{\shortstack{P2p-Gnutella31\\Directed\\$|V|=62,586$\\$|E|=147,892$}} & 0.005 & 100.06 & 1.78 & 4.27 & 99.949 & 0.041 & 0.010 & 81,507 & 4.07 & 38.43 & 0.58 & 1.60\\
 & 0.010 & 26.05 & 6.85 & 4.13 & 99.861 & 0.103 & 0.036 & 21,315 & 3.90 & 65.76 & 1.15 & 3.13\\
 & 0.015 & 11.91 & 14.98 & 4.03 & 99.772 & 0.154 & 0.074 & 9,975 & 3.70 & 109.10 & 1.63 & 4.51\\
 & 0.020 & 7.11 & 25.09 & 3.87 & 99.688 & 0.191 & 0.121 & 5,840 & 3.55 & 130.33 & 2.15 & 6.12\\
 & 0.025 & 4.84 & 36.85 & 3.62 & 99.607 & 0.220 & 0.174 & 3,905 & 3.40 & 171.93 & 2.52 & 7.43\\
 & 0.030 & 3.41 & 52.38 & 3.66 & 99.495 & 0.262 & 0.243 & 2,810 & 3.28 & 236.36 & 2.86 & 8.70\\
\midrule
\multirow{5}{*}{\shortstack{Email-Enron\\Undirected\\$|V|=36,682$\\$|E|=183,831$}} %& 0.005 & 794.51 & 0.30 & 1.12 & 99.995 & 0.004 & 0.001 & 260,923 & 1.12 & 72.22 & 0.25 & 1.26\\
 & 0.010 & 202.43 & 1.18 & 1.10 & 99.984 & 0.013 & 0.003 & 66,882 & 1.09 & 145.51 & 0.48 & 2.46\\
 & 0.015 & 91.36 & 2.63 & 1.09 & 99.970 & 0.024 & 0.006 & 30,236 & 1.07 & 253.06 & 0.71 & 3.62\\
 & 0.020 & 53.50 & 4.48 & 1.05 & 99.955 & 0.035 & 0.010 & 17,676 & 1.03 & 290.30 & 0.93 & 4.83\\
 & 0.025 & 31.99 & 7.50 & 1.11 & 99.932 & 0.052 & 0.016 & 10,589 & 1.10 & 548.22 & 1.21 & 6.48\\
 & 0.030 & 24.06 & 9.97 & 1.03 & 99.918 & 0.061 & 0.021 & 7,923 & 1.02 & 477.32 & 1.38 & 7.34\\
\midrule
\multirow{5}{*}{\shortstack{Cit-HepPh\\Undirected\\$|V|=34,546$\\
$|E|=421,578$}} %& 0.005 & 827.47 & 0.62 & 2.30 & 99.989 & 0.010 & 0.001 & 123,636 & 2.36 & 73.96 & 0.89 & 1.71\\
 & 0.010 & 215.98 & 2.36 & 2.21 & 99.966 & 0.030 & 0.004 & 32,469 & 2.25 & 129.08 & 1.72 & 3.40\\
 & 0.015 & 98.27 & 5.19 & 2.16 & 99.938 & 0.054 & 0.008 & 14,747 & 2.20 & 226.18 & 2.49 & 5.00\\
 & 0.020 & 58.38 & 8.74 & 2.05 & 99.914 & 0.073 & 0.013 & 8,760 & 2.08 & 246.14 & 3.17 & 6.39\\
 & 0.025 & 37.79 & 13.50 & 2.02 & 99.891 & 0.091 & 0.018 & 5,672 & 2.06 & 289.21 & 3.89 & 7.97\\
 & 0.030 & 27.13 & 18.80 & 1.95 & 99.869 & 0.108 & 0.023 & 4,076 & 1.99 & 359.45 & 4.45 & 9.53\\
%Cit-HepTh & 0.005 & 343.91 & 0.13 & 0.84 & 99.988 & 0.010 & 0.002 & 469,999 & 0.79 & 62.09 & 0.39 & 0.96\\
% & 0.010 & 85.23 & 0.52 & 0.92 & 99.966 & 0.028 & 0.006 & 116,185 & 0.80 & 178.79 & 0.79 & 2.26\\
% & 0.015 & 37.75 & 1.17 & 1.07 & 99.941 & 0.047 & 0.012 & 50,420 & 0.82 & 304.23 & 1.18 & 3.23\\
% & 0.020 & 20.69 & 2.13 & 1.28 & 99.912 & 0.067 & 0.021 & 27,453 & 0.85 & 352.04 & 1.61 & 4.71\\
% & 0.025 & 13.40 & 3.29 & 1.53 & 99.727 & 0.226 & 0.047 & 18,111 & 0.82 & 402.29 & 1.94 & 5.29\\
% & 0.030 & 8.53 & 5.16 & 1.99 & 99.604 & 0.326 & 0.070 & 11,613 & 0.89 & 465.89 & 2.39 & 6.44\\
\bottomrule
	\end{tabular}
	\caption{Runtime, speedup, breakdown of runtime, sample size, reduction, and
	absolute error}
	\label{tab:bigtable}
\end{table*}

\subsection{Runtime and Speedup}
Our main goal was to develop an algorithm that can compute an
$(\varepsilon,\delta)$-approximation of the \BC of all nodes as fast as
possible. Hence we evaluate the runtime and the speedup of \staticalgo
w.r.t.~\textsf{BA} and \textsf{RK}. The results are reported in columns 3 to 5
of \Cref{tab:bigtable} (the values for $\varepsilon=0.005$ are missing for
Email-Enron and Cit-HepPh because in these case both \textsf{RK} and \staticalgo
were slower than \textsf{BA}). As expected, the runtime is a perfect linear
function of the sample size (column 9), which in turns grows as
$\varepsilon^{-2}$. The speedup w.r.t.~the exact algorithm \textsf{BA} is
significant and naturally decreases quadratically with $\varepsilon$. More
interestingly \staticalgo is always faster than \textsf{RK}, sometimes by a
significant factor. At first, one may think that this is due to the reduction in
the sample size (column 10), but a deeper analysis shows that this is only one
component of the speedup, which almost always greater than the reduction in
sample size. The other component can be explained by the fact that \textsf{RK}
must perform an expensive computation (computing the
vertex-diameter~\citep{RiondatoK15} of the graph) to determine the sample size
before it can start sampling, while \staticalgo can immediately start sampling
and rely on the stopping condition (whose computation is inexpensive, as we will
discuss). The different speedups for different graphs are due to different
characteristics of the graphs: when the SP DAG between two nodes has many paths,
\staticalgo does more work per sample than \textsf{RK} (which only explore a
single SP on the DAG), hence the speedup is smaller.

\paragraph*{Runtime breakdown} The main challenge in designing a stopping
condition for progressive
sampling algorithm is striking the right balance between the strictness of the
condition (i.e., it should stop early) and the efficiency in evaluating it. We
now comment on the efficiency, and will report about the strictness in
\Cref{sec:expersamplesched,sec:experaccuracy}. In columns 6 to 8 of
\Cref{tab:bigtable} we report the breakdown of the runtime into the main
components. It is evident that evaluating the stopping condition amounts to an
insignificant fraction of the runtime, and most of the time is spent in
computing the samples (selection of nodes, execution of SP algorithm, update of
the \BC estimations). The amount in the ``Other'' column corresponds to time
spent in logging and checking invariants. We can then say that our
stopping condition is extremely efficient to evaluate, and \staticalgo is almost
always doing ``real'' work to improve the estimation.

\subsection{Sample Size and Sample Schedule}\label{sec:expersamplesched}
We evaluate the final sample size of \staticalgo and the performances of the
``automatic'' sample schedule (\Cref{sec:samplesched}). The results are reported
in columns 9 and 10 of \Cref{tab:bigtable}. As expected, the sample size grows
with $\varepsilon^{-2}$. We already commented on the fact that \staticalgo uses
a sample size that is consistently (up to $4\times$) smaller than the one used
by \textsf{RK} and how this is part of the reason why \staticalgo is much faster
than \textsf{RK}. In \Cref{fig:schedules} we show the behavior (on
P2p-Gnutella31, figures for other graphs can be found in
\ifappendix
\Cref{app:exper})
\else
Appendix C of the extended online version~\citep{RiondatoU16ext})
\fi
of the final sample size chosen by the automatic sample schedule in
comparison with \emph{static geometric sample schedules}, i.e., schedules for
which the sample size at iteration $i+1$ is $c$ times the size of the sample
size at iteration $i$. We can see that the \emph{automatic sample schedule is
always better than the geometric ones}, sometimes significantly depending on the
value of $c$ (e.g., more than $2\times$ decrease w.r.t.~using $c=3$ for
$\varepsilon=0.05$). Effectively this means that the automatic sample schedule
really frees the end user from having to selecting a parameter whose impact on
the performances of the algorithm may be devastating (larger final sample size
implies higher runtime). Moreover, we noticed that with the automatic sample
schedule \staticalgo always terminated after just two iterations, while this was
not the case for the geometric sample schedules (taking even 5 iterations in
some cases): this means that effectively the automatic sample schedules
``jumps'' directly to a sample size for which the stopping condition will be
verified. We can then sum up the results and say that the stopping condition
of \staticalgo stops at small sample sizes, smaller than those used in
\textsf{RK} and the automatic sample schedule we designed is extremely efficient
at choosing the right successive sample size, to the point that \staticalgo only
needs two iterations.

\begin{figure}[ht]
	\centering
	\includegraphics[width=\ifappendix.5\else.93\fi\columnwidth]{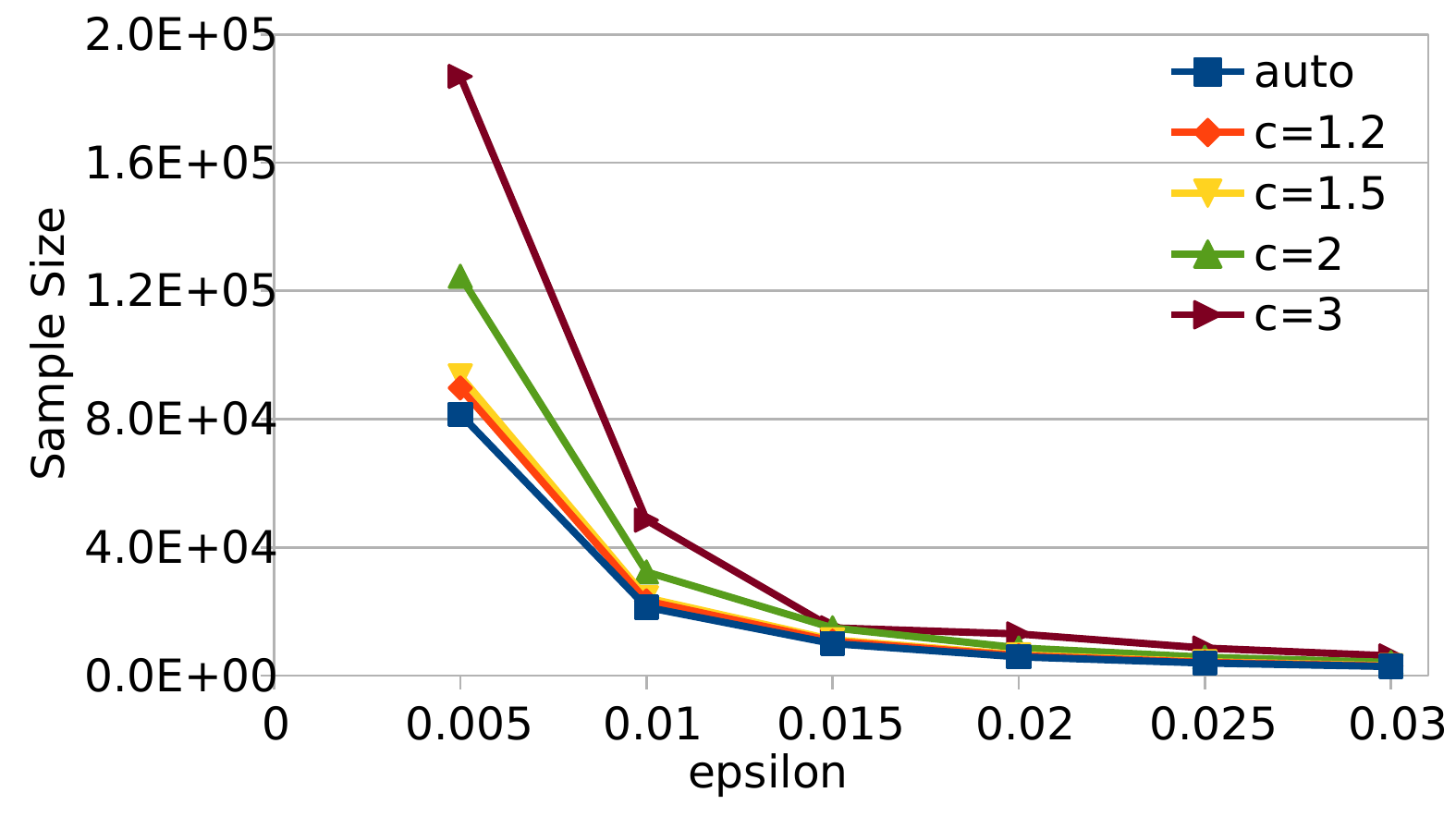}
	\caption{Final sample size for different sample schedules on P2p-Gnutella}
	\label{fig:schedules}
\end{figure}

\begin{figure}[ht]
	\centering
	\includegraphics[width=\ifappendix.5\else.93\fi\columnwidth]{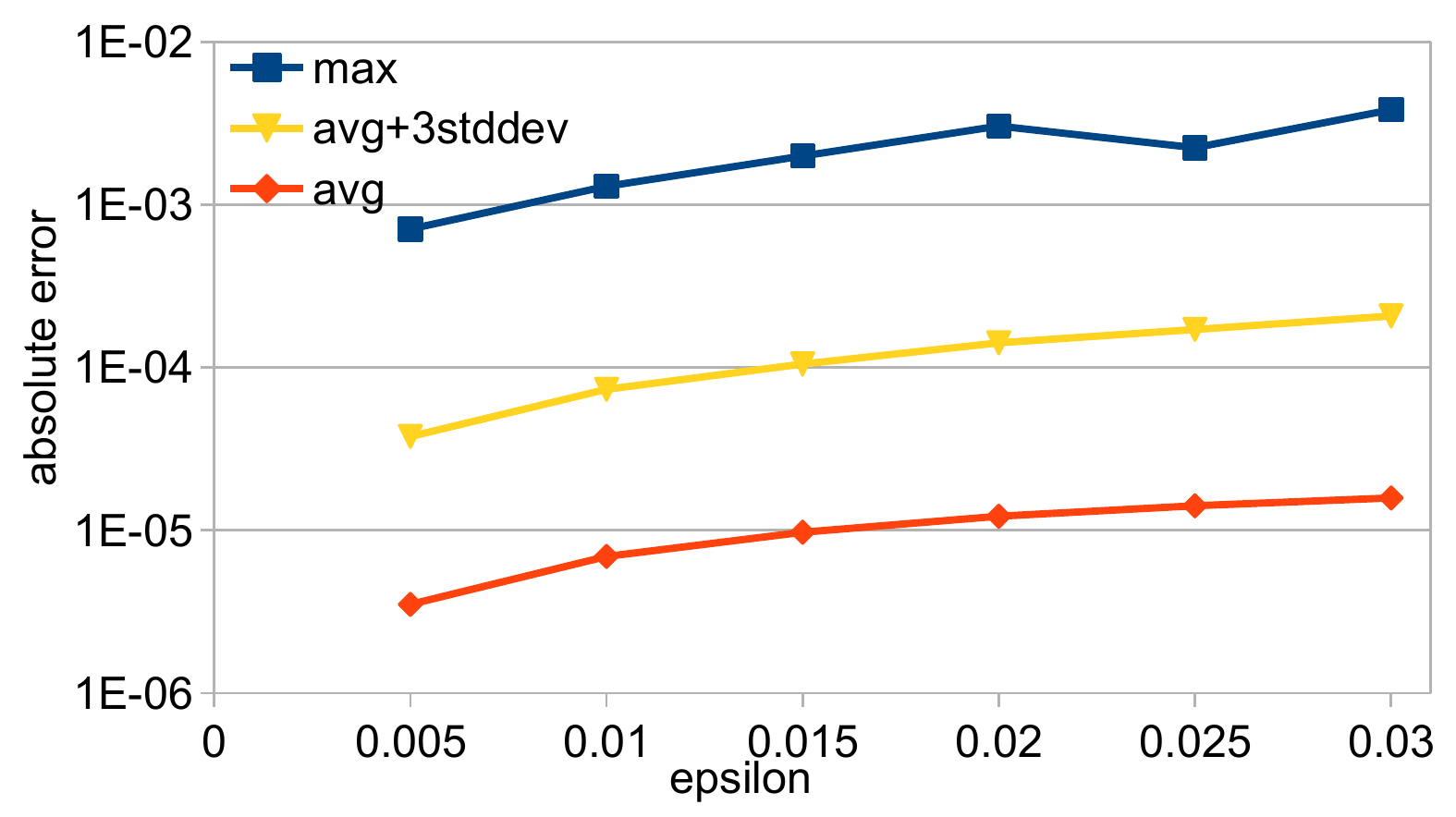}
	\ifappendix
	\else
	\vspace{-5pt}
	\fi
	\caption{Absolute error evaluation -- Soc-Epinions1}
	\label{fig:errors}
\end{figure}

\subsection{Accuracy}\label{sec:experaccuracy}
We evaluate the accuracy of \staticalgo by measuring the
absolute error $|\tbetw(v)-\betw(v)|$. The theoretical analysis
guarantees that this quantity should be at most $\varepsilon$ for all nodes,
with probability at least $1-\delta$. A first important result is that in
\emph{all} the thousands of runs of \staticalgo, the maximum error was
\emph{always} smaller than $\varepsilon$ (not just with probability
$>1-\delta$). We report statistics about the absolute error in the three
rightmost columns of \Cref{tab:bigtable} and in~\Cref{fig:errors} (figures for
the other graphs are in
\ifappendix
\Cref{app:exper}.
\else
Appendix C of the extended online version~\citep{RiondatoU16ext}.
\fi
The minimum error (not reported) was always 0. %so we do not report it.
The maximum error is \emph{an order of magnitude smaller than $\varepsilon$}, and
the average error is around \emph{three orders of magnitude} smaller than
$\varepsilon$, with a very small standard deviation. As expected, the error
grows as $\varepsilon^{-2}$. In \Cref{fig:errors} we
show the behavior of the maximum, average, and average plus three standard
deviations (approximately corresponding to the 95\% percentile) for
Soc-Epinions1 (the vertical axis has a logarithmic scale), to appreciate how
most of the errors are almost two orders of magnitude smaller than
$\varepsilon$.

All these results show that \emph{\staticalgo is very accurate, more than what
is guaranteed by the theoretical analysis}. This can be explained by the
fact that the bounds to the sampling size, the stopping condition, and the
sample schedule are \emph{conservative}, in the sense that we may be sampling
more than necessary to obtain an $(\varepsilon,\delta)$-approximation.
Tightening any of these components would result in a less conservative algorithm
that still offers the same approximation quality guarantees, and is an
interesting research direction.

\subsection{Dynamic BC Approximation}
We did not evaluate \dynamicalgo experimental, but, given its design, one can
expect that, when compared to previous contributions offering the same quality
guarantees~\citep{BergaminiM15arXiv,HayashiAY15}, it would exhibit similar or
even larger speedups and reduction in the sample size than what \staticalgo had
w.r.t.~\textsf{RK}. Indeed, the algorithm by~\citet{BergaminiM15} uses
\textsf{RK} as a building block and it needs to constantly keep track of (an
upper bound to) the vertex diameter of the graph, a very expensive operation. On
the other hand, the analysis of the sample size by~\citet{HayashiAY15} uses very
loose simultaneous deviation bounds (the union bound). As already shown
by~\citet{RiondatoK15}, the resulting sample size is extremely large and
they already showed how \textsf{RK} can use a smaller sample size. Since we
built over the work by~\citet{HayashiAY15} and \staticalgo improves over
\textsf{RK}, we can reasonably expect it to have much better performances than
the algorithm by~\citet{HayashiAY15}
% 20160209 Matteo: all the following go to the journal version
%
%\subsection{Scalability}
%\todo{Generate graphs of different size using a random graph generator
%(Barabasi-Albert and see what happens. We want to show:
	% Scalability as function of $|V|$ on artificial graphs (with breaking up of running time?)
	%Scalability as function of $|V|$ on artificial graph, comparison with VC.
%}
%
%
%\subsubsection*{Accuracy}
%
%\subsubsection*{Runtime and Scalability}

\section{Conclusions}\label{sec:concl}
We presented \textsf{\algobase}, a family of sampling-based %randomized
algorithms for computing and maintaining high-quality approximations
of (variants of) the \BC of all vertices in a graph. Our
algorithms can handle static and dynamic graphs with edge updates (both
deletions and insertions). We discussed a number of variants of our basic
algorithms, including finding the top-$k$ nodes with higher \BC, using improved
estimators, and special cases when there is a single SP. \textsf{\algobase{}}
greatly improves, theoretically and experimentally, the current state of the
art. The analysis relies on Rademacher averages and on pseudodimension. %, fundamental concepts from statistical learning theory.
To our knowledge this is the first application of these concepts to
graph mining.

In the future we plan to investigate stronger bounds to the Rademacher averages,
give stricter bounds to the sample complexity of \BC by studying the
pseudodimension of the class of functions associated to it, and extend our study
to other network measures.

\ifappendix
\para{Acknowledgements} The authors are thankful to Elisabetta Bergamini and
Christian Staudt for their help with the NetworKit code.

This work was supported in part by NSF grant IIS-1247581 and NIH grant
R01-CA180776.
\fi

\ifacmstyle
%\newpage % We have 10 pages not including rerfences, which can then go to a new page.
\fi
\ifbibtex
\bibliographystyle{abbrvnat}
\bibliography{centrality,riondapubs,various,vcmine}

\begin{thebibliography}{36}
\providecommand{\natexlab}[1]{#1}
\providecommand{\url}[1]{\texttt{#1}}
\expandafter\ifx\csname urlstyle\endcsname\relax
  \providecommand{\doi}[1]{doi: #1}\else
  \providecommand{\doi}{doi: \begingroup \urlstyle{rm}\Url}\fi

\bibitem[Anguita et~al.(2014)Anguita, Ghio, Oneto, and Ridella]{AnguitaGOR14}
D.~Anguita, A.~Ghio, L.~Oneto, and S.~Ridella.
\newblock A deep connection between the {V}apnik-{C}hervonenkis entropy and the
  {R}ademacher complexity.
\newblock \emph{IEEE Transactions on Neural Networks and Learning Systems},
  25\penalty0 (12):\penalty0 2202--2211, 2014.

\bibitem[Anthonisse(1971)]{Anthonisse71}
J.~M. Anthonisse.
\newblock The rush in a directed graph.
\newblock Technical Report BN 9/71, Stichting Mathematisch Centrum, Amsterdam,
  Netherlands, 1971.

\bibitem[Anthony and Bartlett(1999)]{AnthonyB99}
M.~Anthony and P.~L. Bartlett.
\newblock \emph{Neural Network Learning - Theoretical Foundations}.
\newblock Cambridge University Press, New York, NY, USA, 1999.
\newblock ISBN 978-0-521-57353-5.

\bibitem[Anthony and Shawe-Taylor(1993)]{AnthonyST93}
M.~Anthony and J.~Shawe-Taylor.
\newblock A result of {V}apnik with applications.
\newblock \emph{Discrete Applied Mathematics}, 47\penalty0 (3):\penalty0
  207--217, 1993.

\bibitem[Bader et~al.(2007)Bader, Kintali, Madduri, and Mihail]{BaderKMM07}
D.~A. Bader, S.~Kintali, K.~Madduri, and M.~Mihail.
\newblock Approximating betweenness centrality.
\newblock In A.~Bonato and F.~Chung, editors, \emph{Algorithms and Models for
  the Web-Graph}, volume 4863 of \emph{Lecture Notes in Computer Science},
  pages 124--137. Springer Berlin Heidelberg, 2007.
\newblock ISBN 978-3-540-77003-9.
\newblock \doi{10.1007/978-3-540-77004-6_10}.

\bibitem[Bartlett and Lugosi(1999)]{BartlettL99}
P.~L. Bartlett and G.~Lugosi.
\newblock An inequality for uniform deviations of sample averages from their
  means.
\newblock \emph{Statistics {\&} Probability Letters}, 44\penalty0 (1):\penalty0
  55--62, 1999.

\bibitem[Bergamini and Meyerhenke(2015{\natexlab{a}})]{BergaminiM15}
E.~Bergamini and H.~Meyerhenke.
\newblock Fully-dynamic approximation of betweenness centrality.
\newblock \emph{CoRR}, abs/1504.0709 (to appear in ESA'15), Apr.
  2015{\natexlab{a}}.

\bibitem[Bergamini and Meyerhenke(2015{\natexlab{b}})]{BergaminiM15arXiv}
E.~Bergamini and H.~Meyerhenke.
\newblock Approximating betweenness centrality in fully-dynamic networks.
\newblock \emph{CoRR}, abs/1510.07971, Oct 2015{\natexlab{b}}.
\newblock URL \url{http://arxiv.org/abs/1510.07971}.

\bibitem[Bergamini et~al.(2015)Bergamini, Meyerhenke, and
  Staudt]{BergaminiMS15}
E.~Bergamini, H.~Meyerhenke, and C.~L. Staudt.
\newblock Approximating betweenness centrality in large evolving networks.
\newblock In \emph{17th Workshop on Algorithm Engineering and Experiments,
  ALENEX 2015}, pages 133--146. SIAM, 2015.

\bibitem[Boucheron et~al.(2005)Boucheron, Bousquet, and Lugosi]{BoucheronBL05}
S.~Boucheron, O.~Bousquet, and G.~Lugosi.
\newblock Theory of classification : A survey of some recent advances.
\newblock \emph{{ESAIM}: Probability and Statistics}, 9:\penalty0 323--375,
  2005.

\bibitem[Boyd and Vandenberghe(2004)]{BoydV04}
S.~Boyd and L.~Vandenberghe.
\newblock \emph{Convex optimization}.
\newblock Cambridge university press, 2004.

\bibitem[Brandes(2001)]{Brandes01}
U.~Brandes.
\newblock A faster algorithm for betweenness centrality.
\newblock \emph{J. Math. Sociol.}, 25\penalty0 (2):\penalty0 163--177, 2001.
\newblock \doi{10.1080/0022250X.2001.9990249}.

\bibitem[Brandes and Pich(2007)]{BrandesP07}
U.~Brandes and C.~Pich.
\newblock Centrality estimation in large networks.
\newblock \emph{Int. J. Bifurcation and Chaos}, 17\penalty0 (7):\penalty0
  2303--2318, 2007.
\newblock \doi{10.1142/S0218127407018403}.

\bibitem[Cortes et~al.(2013)Cortes, Greenberg, and Mohri]{CortesGM13}
C.~Cortes, S.~Greenberg, and M.~Mohri.
\newblock Relative deviation learning bounds and generalization with unbounded
  loss functions.
\newblock \emph{CoRR}, abs/1310.5796, Oct 2013.
\newblock URL \url{http://arxiv.org/abs/1310.5796}.

\bibitem[Erd\H{o}s et~al.(2015)Erd\H{o}s, Ishakian, Bestavros, and
  Terzi]{ErdosIBT15}
D.~Erd\H{o}s, V.~Ishakian, A.~Bestavros, and E.~Terzi.
\newblock A divide-and-conquer algorithm for betweenness centrality.
\newblock In \emph{SIAM Data Mining Conf.}, 2015.

\bibitem[Freeman(1977)]{Freeman77}
L.~C. Freeman.
\newblock A set of measures of centrality based on betweenness.
\newblock \emph{Sociometry}, 40:\penalty0 35--41, 1977.

\bibitem[Geisberger et~al.(2008)Geisberger, Sanders, and
  Schultes]{GeisbergerSS08}
R.~Geisberger, P.~Sanders, and D.~Schultes.
\newblock Better approximation of betweenness centrality.
\newblock In J.~I. Munro and D.~Wagner, editors, \emph{Algorithm Eng.~\&
  Experiments (ALENEX'08)}, pages 90--100. SIAM, 2008.

\bibitem[Green et~al.(2012)Green, McColl, and Bader]{GreenMB12}
O.~Green, R.~McColl, and D.~Bader.
\newblock A fast algorithm for streaming betweenness centrality.
\newblock In \emph{Privacy, Security, Risk and Trust (PASSAT), 2012
  International Conference on and 2012 International Confernece on Social
  Computing (SocialCom)}, pages 11--20, sep 2012.
\newblock \doi{10.1109/SocialCom-PASSAT.2012.37}.

\bibitem[Har-Peled and Sharir(2011)]{HarPS11}
S.~Har-Peled and M.~Sharir.
\newblock Relative $(p,\varepsilon)$-approximations in geometry.
\newblock \emph{Discrete \& Computational Geometry}, 45\penalty0 (3):\penalty0
  462--496, 2011.
\newblock ISSN 0179-5376.
\newblock \doi{10.1007/s00454-010-9248-1}.

\bibitem[Haussler(1992)]{Haussler92}
D.~Haussler.
\newblock Decision theoretic generalizations of the {PAC} model for neural net
  and other learning applications.
\newblock \emph{Information and Computation}, 100\penalty0 (1):\penalty0
  78--150, 1992.
\newblock ISSN 0890-5401.

\bibitem[Hayashi et~al.(2015)Hayashi, Akiba, and Yoshida]{HayashiAY15}
T.~Hayashi, T.~Akiba, and Y.~Yoshida.
\newblock Fully dynamic betweenness centrality maintenance on massive networks.
\newblock \emph{Proceedings of the VLDB Endowment}, 9\penalty0 (2), 2015.

\bibitem[Kas et~al.(2013)Kas, Wachs, Carley, and Carley]{KasWCC13}
M.~Kas, M.~Wachs, K.~M. Carley, and L.~R. Carley.
\newblock Incremental algorithm for updating betweenness centrality in
  dynamically growing networks.
\newblock In \emph{Proceedings of the 2013 IEEE/ACM International Conference on
  Advances in Social Networks Analysis and Mining}, ASONAM '13, pages 33--40,
  New York, NY, USA, 2013. ACM.
\newblock ISBN 978-1-4503-2240-9.
\newblock \doi{10.1145/2492517.2492533}.
\newblock URL \url{http://doi.acm.org/10.1145/2492517.2492533}.

\bibitem[Kourtellis et~al.(2015)Kourtellis, Morales, and
  Bonchi]{KourtellisMB15}
N.~Kourtellis, G.~D.~F. Morales, and F.~Bonchi.
\newblock Scalable online betweenness centrality in evolving graphs.
\newblock \emph{{IEEE} Trans. Knowl. Data Eng.}, 27\penalty0 (9):\penalty0
  2494--2506, 2015.
\newblock \doi{10.1109/TKDE.2015.2419666}.

\bibitem[Lee et~al.(2012)Lee, Lee, Park, Choi, and Chung]{LeeLPCC21}
M.-J. Lee, J.~Lee, J.~Y. Park, R.~H. Choi, and C.-W. Chung.
\newblock {QUBE}: A quick algorithm for updating betweenness centrality.
\newblock In \emph{Proceedings of the 21st International Conference on World
  Wide Web}, WWW '12, pages 351--360, New York, NY, USA, 2012. ACM.
\newblock ISBN 978-1-4503-1229-5.
\newblock \doi{10.1145/2187836.2187884}.

\bibitem[Leskovec and Krevl(2014)]{LeskovecK14}
J.~Leskovec and A.~Krevl.
\newblock {SNAP Datasets}: {Stanford} large network dataset collection.
\newblock \url{http://snap.stanford.edu/data}, June 2014.

\bibitem[Li et~al.(2001)Li, Long, and Srinivasan]{LiLS01}
Y.~Li, P.~M. Long, and A.~Srinivasan.
\newblock Improved bounds on the sample complexity of learning.
\newblock \emph{J. Comp. Sys. Sci.}, 62\penalty0 (3):\penalty0 516--527, 2001.
\newblock ISSN 0022-0000.
\newblock \doi{10.1006/jcss.2000.1741}.

\bibitem[L\"{o}ffler and Phillips(2009)]{LofflerP09}
M.~L\"{o}ffler and J.~M. Phillips.
\newblock Shape fitting on point sets with probability distributions.
\newblock In A.~Fiat and P.~Sanders, editors, \emph{Algorithms - ESA 2009},
  volume 5757 of \emph{Lecture Notes in Computer Science}, pages 313--324.
  Springer Berlin Heidelberg, 2009.
\newblock \doi{10.1007/978-3-642-04128-0_29}.

\bibitem[Newman(2010)]{Newman10}
M.~E.~J. Newman.
\newblock \emph{Networks -- An Introduction}.
\newblock Oxford University Press, 2010.

\bibitem[Oneto et~al.(2013)Oneto, Ghio, Anguita, and Ridella]{OnetoGAR13}
L.~Oneto, A.~Ghio, D.~Anguita, and S.~Ridella.
\newblock An improved analysis of the {R}ademacher data-dependent bound using
  its self bounding property.
\newblock \emph{Neural Networks}, 44:\penalty0 107--111, 2013.

\bibitem[Pollard(1984)]{Pollard84}
D.~Pollard.
\newblock \emph{Convergence of stochastic processes}.
\newblock Springer-Verlag, 1984.

\bibitem[Riondato and Kornaropoulos(2015)]{RiondatoK15}
M.~Riondato and E.~M. Kornaropoulos.
\newblock Fast approximation of betweenness centrality through sampling.
\newblock \emph{Data Mining and Knowledge Discovery}, 30\penalty0 (2):\penalty0
  438--475, 2015.
\newblock ISSN 1573-756X.
\newblock \doi{10.1007/s10618-015-0423-0}.
\newblock URL \url{http://dx.doi.org/10.1007/s10618-015-0423-0}.

\bibitem[Riondato and Upfal(2015)]{RiondatoU15}
M.~Riondato and E.~Upfal.
\newblock Mining frequent itemsets through progressive sampling with
  {R}ademacher averages.
\newblock In \emph{Proc.~21st ACM SIGKDD Int.~Conf.~Knowl.~Disc.~and Data
  Mining}, 2015.
\newblock URL
  \url{http://matteo.rionda.to/papers/RiondatoUpfal-FrequentItemsetsSamplingRademacher-KDD.pdf}.
\newblock Extended Version.

\bibitem[Sar{\i}y\"{u}ce et~al.(2013)Sar{\i}y\"{u}ce, Saule, Kaya, and
  \c{C}ataly\"{u}rek]{SaryuceSKC13}
A.~E. Sar{\i}y\"{u}ce, E.~Saule, K.~Kaya, and U.~V. \c{C}ataly\"{u}rek.
\newblock Shattering and compressing networks for betweenness centrality.
\newblock In \emph{SIAM Data Mining Conf.}, 2013.

\bibitem[Shalev-Shwartz and Ben-David(2014)]{ShalevSBD14}
S.~Shalev-Shwartz and S.~Ben-David.
\newblock \emph{Understanding Machine Learning: From Theory to Algorithms}.
\newblock Cambridge University Press, 2014.

\bibitem[Staudt et~al.(2014)Staudt, Sazonovs, and Meyerhenke]{StaudtSM14}
C.~Staudt, A.~Sazonovs, and H.~Meyerhenke.
\newblock {N}etwor{K}it: An interactive tool suite for high-performance network
  analysis.
\newblock \emph{CoRR}, abs/1403.3005, March 2014.

\bibitem[Vapnik(1999)]{Vapnik99}
V.~N. Vapnik.
\newblock \emph{The Nature of Statistical Learning Theory}.
\newblock Statistics for engineering and information science. Springer-Verlag,
  New York, NY, USA, 1999.
\newblock ISBN 9780387987804.

\end{thebibliography}
\else
\input{biblio}
\fi

\ifappendix
\appendix

\section{Relative-error Top-k Approximation}\label{app:topk}
In this section we prove the correctness of the algorithm \staticalgotopk
(\Cref{thm:topk}). The pseudocode can be found in Algorithm~\ref{alg:topk}.

\begin{proof}[of~\Cref{thm:topk}]
	With probability at least $1-\delta'$, the set $\tilde{B}'$ computed during
	the first phase (execution of \staticalgo) has the properties \ldots . With
	probability at least $1-\delta''$, the set $\tilde{B}''$ computed during the
	second phase (execution of \staticalgo) has the properties
	from~\Cref{thm:staticalgorel}. Suppose both these events occur, which
	happens with probability at least $1-\delta$. Consider the value $\ell'$. It
	is straightforward to check that $\ell'$ is a lower bound to $b_k$: indeed
	there must be at least $k$ nodes with exact \BC at least $\ell'$. For the
	same reasons, and considering the fact that we run \staticalgorel with
	parameters $\varepsilon$, $\delta''$, and $\lambda=\ell'$, we have that
	$\ell''\le b_k$. From this and the definition of $\widetilde{\TOP}(k,G)$, it
	follows that the elements of $\widetilde{\TOP}(k,G)$ are such that their
	exact may be greater than $\ell''$, and therefore of $b_k$. This means that
	$\TOP(k,G)\subseteq\widetilde{TOP}(k,G)$. The other properties of
	$\widetilde{TOP}(k,G)$ follow from the properties of the output of
	\staticalgorel.
\end{proof}

\begin{algorithm}[ht]
	\DontPrintSemicolon
	\SetKwInOut{Input}{input}
	\SetKwInOut{Output}{output}
	\SetKwFunction{GetSample}{uniform\_random\_sample}
	\SetKwFunction{ModifiedSP}{compute\_SPs}
	\SetKwFunction{UpdateDictionaries}{update\_dicts}
	\SetKwComment{tcp}{//}{}
	\Input{Graph $G=(V,E)$, accuracy parameter $\varepsilon\in(0,1)$, confidence parameter $\delta\in(0,1)$, value $k\ge 1$}
	\Output{Set $\widetilde{B}$ of approximations of the \BC of the top-$k$ vertices in $V$ with highest \BC}
	$\delta',\delta''\leftarrow$ reals such that $(1-\delta_1)(1-\delta_2)\ge 1-\delta$\;
	$\widetilde{B}'\leftarrow$ output of $\staticalgo$ run with input $G,
		\varepsilon, \delta'$\;
	$\ell'\leftarrow$ $k$-th highest value in $\widetilde{B}'$\;
	$\tilde{b}'=\ell'-\varepsilon$\;
	$\widetilde{B}\leftarrow$ output of a variant of $\staticalgo$ using the
		definition of $\Delta_i$ from~\eqref{eq:deltairel}, and input $G,
		\varepsilon, \delta'', \tilde{b}'$\;
	\Return{$\widetilde{B}$}\;
	\caption{\staticalgotopk: relative-error approximation of top-$k$ \BC nodes on static graph}
	\label{alg:topk}
\end{algorithm}

\section{Special Cases}\label{app:unique}
In this section we expand on our discussion from~\Cref{sec:unique}. Since our
results rely on pseudodimension~\citep{Pollard84}, we start with a presentation
of the fundamental definitions and results about pseudodimension.

\subsection{Pseudodimension}
Before introducing the pseudodimension, we must recall some notions and
results about the Vapnik-Chervonenkis (VC) dimension. We refer the reader to the
books by~\citet{ShalevSBD14} and by~\citet{AnthonyB99} for an in-depth
exposition of VC-dimension and pseudodimension.

Let $D$ be a domain and let $\range$ be a collection of subsets of $D$
($\range\subseteq 2^D$). We call $\range$ a \emph{rangeset on $D$}. Given
$A\subseteq D$, the \emph{projection of $\range$ on $A$} is $P_\range(A)=\{R\cap
A ~:~ R\in\range\}$.  When $P_\range(A)=2^A$, we say that $A$ is
\emph{shattered} by $\range$. Given $B\subseteq D$, the \emph{empirical
VC-dimension} of $\range$, denoted as $\EVC(\range, B)$ is the size of the
largest subset of $B$ that can be shattered. The \emph{VC-dimension} of
$\range$, denoted as $\VC(\range)$ is defined as $\VC(\range)=\EVC(\range,D)$.

Let $\family$ be a class of functions from some domain $D$ to $[0,1]$. Consider,
for each $f\in\family$, the subset $R_f$ of $D\times[0,1]$ defined as
\[
R_f=\{(x,t) ~:~ t\le f(x)\}\enspace.
\]
We define a rangeset $\family^+$ on $D\times[0,1]$ as $
\family^+ = \{R_{f}, f\in\family\}$.
The \emph{empirical pseudodimension}~\citep{Pollard84} of $\family$ on a subset
$B\subseteq D$, denoted as
$\EPD(\family, B)$, is the empirical VC-dimension of $\family^+$:
$\EPD(\family, B)=\EVC(\family^+, B)$. The pseudodimension of $\family$, denoted
as $\PD(\family)$ is the VC-dimension of $\family^+$,
$\PD(\family)=\VC(\family^+)$~\citep[Sect.~11.2]{AnthonyB99}.
Having an upper bound to the pseudodimension of $\family$ allows to
bound the supremum of the deviations from~\eqref{eq:supdev}, as stated in the
following result.

\begin{theorem}[\citep{LiLS01}, see also~\citep{HarPS11}]\label{thm:eapprox}
  Let $D$ be a domain and $\family$ be a family of functions from $D$ to
  $[0,1]$. Let $\PD(\family)\le d$. Given $\varepsilon,\delta\in(0,1)$,  let
  $\Sam$ be a collection of elements sampled independently and uniformly at random
  from $D$, with size
  \begin{equation}\label{eq:eapprox}
    |\Sam|=\frac{c}{\varepsilon^2}\left(d+\log\frac{1}{\delta}\right)\enspace.
  \end{equation}
  Then
  \[
  \Pr\left(\exists f\in\family \mbox{ s.t. }
  \left|\mean_D(f)-\mean_\Sam(f)\right|>\varepsilon\right)<\delta\enspace.
  \]
\end{theorem}
The constant $c$ is universal and it is less than 0.5~\citep{LofflerP09}.

The following two technical lemmas are, to the best of our knowledge, new.
We use them later to bound the pseudodimension of a family of functions related
to betweenness centrality.

\begin{lemma}\label{lem:uniqueelement}
  Let $B\subseteq D\times[0,1]$ be a set that is shattered by $\family^+$.
  Then $B$ can contain \emph{at most one} element $(d,x)\in D\times[0,1]$ for
  each $d\in D$.
\end{lemma}
\begin{proof}
  Let $d\in D$ and consider any two distinct values $x_1,x_2\in[0,1]$. Let,
  w.l.o.g., $x_1< x_2$ and let
  $B=\{(\tau,x_1),(\tau,x_2)\}$. From the definitions of the ranges, there is no
  $R\in\family^+$ such that $R\cap B=\{(d,x_1)\}$, therefore $B$ can not be
  shattered, and so neither can any of its supersets, hence the thesis.
\end{proof}

\begin{lemma}\label{lem:nozero}
  Let $B\subseteq D\times[0,1]$ be a set that is shattered by $\family^+$. Then
  $B$ does not contain \emph{any} element in the form $(d,0)$, for any $d\in D$.
\end{lemma}
\begin{proof}
	For any $d\in D$, $(d,0)$ is contained in every $R\in\family^+$,
	hence given a set $B=\{(d,0)\}$ it is impossible to find a range
	$R_\emptyset$ such that $B\cap R_\emptyset = \emptyset$, therefore $B$ can
	not be shattered, nor can any of its supersets, hence the thesis.
\end{proof}

\subsection{Pseudodimension for BC}
We now move to proving the results in~\Cref{sec:unique}.

Let $G=(V,E)$ be a graph, and consider the family
\[
	\family=\{f_w, w\in V\}
\]
where $f_w$ goes from $\domain=\{(u,v)\in V\times V, u\neq v\}$ to $[0,1]$ and
is defined in~\eqref{eq:functionf}. The rangeset $\family^+$ contains one range
$R_w$ for each node $w\in V$. The set $R_w\subseteq\domain\times[0,1]$ contains
pairs in the form $((u,v), x)$, with $(u,v)\in\domain$ and $x\in[0,1]$. The
pairs $((u,v), x)\in R_w$ with $x > 0$ are all and only the pairs with this form
such that
\begin{enumerate}
	\item $w$ is on a SP from $u$ to $v$; and
	\item $x \le \sigma_{uv}(w)/\sigma_{uv}$.
\end{enumerate}

We now prove a result showing that some subsets of $\domain\times[0,1]$ can not
be shattered by $\family^+$, on any graph $G$. \Cref{thm:unique} follows
immediately from this result, and \Cref{corol:unique} then follows from
\Cref{thm:unique,thm:eapprox}.

\begin{lemma}\label{lem:unique}
	There exists no undirected graph $G=(V,E)$ such that it is possible to
	shatter a set
	\[
		B=\{((u_i, v_i),x_i), 1\le i\le 4\}\subseteq\domain\times[0,1]
	\]
	if there are at least three distinct values $j',j'',j'''\in[1,4]$ for which
	\[
		\sigma_{u_{j'}v_{j'}}=\sigma_{u_{j''}v_{j''}}=\sigma_{u_{j'''}v_{j'''}}=1\enspace.
	\]
\end{lemma}

\begin{proof}
	First of all, according to~\Cref{lem:uniqueelement,lem:nozero}, for $B$ to
	be shattered it must be
	\[
		(u_i,v_i)\neq (u_j,v_j) \mbox{ for } i\neq j
	\]
	and $x_i\in(0,1]$, $1\le i\le 4$.

	\citet[Lemma 2]{RiondatoK15} showed that there exists no undirected graph
	$G=(V,E)$ such that it is possible to shatter $B$
	if
	\[
		\sigma_{u_1v_1}=\sigma_{u_2v_2}=\sigma_{u_3v_3}=\sigma_{u_4v_4}=1\enspace.
	\]
	Hence, what we need to show to prove the thesis is that it is impossible to
	build an undirected graph $G=(V,E)$ such that $\family^+$ can shatter $B$
	when the elements of $B$ are such that
	\[
		\sigma_{u_1v_1}=\sigma_{u_2v_2}=\sigma_{u_3v_3}=1
	\]
	and $\sigma_{u_4v_4}=2$.

	Assume now that such a graph $G$ exists and therefore $B$ is shattered by
	$\family^+$.

	For $1\le i\le 3$, let $p_i$ be the \emph{unique} SP from $u_i$ to $v_i$,
	and let $p_4'$ and $p_4''$ be the two SPs from $u_4$ to $v_4$.

	First of all, notice that if any two of $p_1$, $p_2$, $p_3$ meet at a node
	$a$ and separate at a node $b$, then they can not meet again at any node
	before $a$ or after $b$, as otherwise there would be multiple SPs
	between their extreme nodes, contradicting the hypothesis. Let this fact be
	denoted as $\mathsf{F}_1$.

	Since $B$ is shattered, its subset
	\[
		A=\{((u_i, v_i),x_i 1\le i\le 3\}\subset B
	\]
	is also shattered, and in particular it can be shattered
	by a collection of ranges that is a subset of a collection of ranges that
	shatters $B$. We now show some facts about the properties of this shattering
	which we will use later in the proof.

	Define
	\[
		i^+=\left\{\begin{array}{ll}
			i+1 & \mbox{if } i=1,2\\
			1 & \mbox{if } i=3
		\end{array}\right.
	\]
	and
	\[
		i^-=\left\{\begin{array}{ll}
			3 & \mbox{if } i=1\\
			i-1 & \mbox{if } i=2,3
		\end{array}\right.\enspace.
	\]

	Let $v_A$ be a node such that $R_{v_A}\cap A=A$.
	For any $i$, $1\le i\le 3$, let $v_{i,i^+}$ be the node such that
	\[
		R_{v_{i,i^+}}\cap A=\{(u_i,v_i),(u_{i^+},v_{i^+}\}\enspace.
	\]
	Analogously, let $v_{i,i^-}$ be the node such that
	\[
		R_{v_{i,i^-}}\cap A=\{(u_i,v_i),(u_{i^-},v_{i^-}\}\enspace.
	\]
	We want to show that $v_A$ is on the SP connecting $v_{i,i^+}$ to $v_{i,i^-}$.
	Assume it was not. Then we would have that either $v_{i,i^+}$ is between
	$v_A$ and $v_{i,i^-}$ or $v_{i,i^-}$ is between $v_A$ and $v_{i,i^+}$. Assume
	it was the former (the latter follows by symmetry). Then
	\begin{enumerate}
		\item there must be a SP $p'$ from $u_{i^-}$ to $v_{i^+}$ that goes
			through $v_{i,i^-}$;
		\item there must be a SP $p''$ from $u_{i^-}$ to $v_{i^+}$ that goes
			through $v_A$;
		\item there is no SP from $u_{i^-}$ to $v_{i^+}$ that goes through
			$v_{i,i^+}$.
	\end{enumerate}
	Since there is only one SP from $u_{i^-}$ to $v_{i^-}$, it must be that
	$p'=p''$. But then $p'$ is a SP that goes through $v_{i,i-}$ and through
	$v_A$ but not through $v_{i,i^+}$, and $p_i$ is a SP that goes through
	$v_{i,i^-}$, through $v_{i,i^+}$ and through $v_A$ (either in this order or
	in the opposite). This means that there are at least two SPs between
	$v_{i,i^-}$ and $v_A$, and therefore there would be two SPs between $u_{i}$
	and $v_{i}$, contradicting the hypothesis that there is only one SP between
	these nodes.  Hence it must be that $v_A$ is between $v_{i,i^-}$ and
	$v_{i,i^+}$. This is true for all $i$, $1\le i\le 3$. Denote this fact as
	$\mathsf{F}_2$.

	Consider now the nodes $v_{i,4}$ and $v_{j,4}$. We now show that they can
	not belong to the same SP from $u_4$ and $v_4$.

	\begin{itemize}
		\item Assume that $v_{i,4}$ and $v_{j,4}$ are on the same SP $p$ from
			$u_4$ to $v_4$ and assume that $v_{i,j,4}$ is also on $p$. Consider
			the possible orderings of $v_{i,4}$, $v_{j,4}$ and $v_{i,j,4}$ along
			$p$.
			\begin{itemize}
				\item If the ordering is $v_{i,4}$, then $v_{j,4}$, then
					$v_{i,j,4}$ or $v_{j,4}$, then $v_{j,4}$, then $v_{i,j,4}$,
					or the reverses of these orderings (for a total of four
					orderings), then it is easy to see that fact $\mathsf{F}_1$
					would be contradicted, as there are two different SPs from
					the first of these nodes to the last, one that goes through
					the middle one, and one that does not, but then there would
					be two SPs between the pair of nodes $(u_k,v_k)$ where $k$
					is the index in $\{1,2,3\}$ different than $4$ that is in
					common between the first and the last nodes in this
					ordering, and this would contradict the hypothesis, so these
					orderings are not possible.
				\item Assume instead the ordering is such that $v_{i,j,4}$ is
					between $v_{i,4}$ and $v_{j,4}$ (two such ordering exist).
					Consider the paths $p_i$ and $p_j$. They must meet at some
					node $v_{f_{i,j}}$ and separate at some node $v_{l_{i,j}}$.
					From the ordering, and fact $\mathsf{F}_1$, $v_{i,j,4}$ must
					be between these two nodes. From fact $\mathsf{F}_2$ we have
					that also $v_{A}$ must be between these two nodes. Moreover,
					neither $v_{i,4}$ nor $v_{j,4}$ can be between these two
					nodes. But then consider the SP $p$. This path must go
					together with $p_i$ (resp.~$p_j$) from at least $p_{i,4}$
					(resp.~$p_{j,4}$) to the farthest between $v_{f_{i,j}}$ and
					$v_{l_{i,j}}$ from $p_{i,4}$ (resp.~$p_{j,4}$). Then in
					particular $p$ goes through all nodes between $v_{f_{i,j}}$
					and $v_{l_{i,j}}$ that $p_i$ and $p_j$ go through. But since
					$v_A$ is among these nodes,and $v_A$ can not belong to $p$,
					this is impossible, so these orderings of the nodes $v_{i,4}$,
					$v_{j,4}$, and $v_{i,j,4}$ are not possible.
			\end{itemize}
			Hence we showed that $v_{i,4}$, $v_{j,4}$, and $v_{i,j,4}$ can not
			be on the same SP from $u_4$ to $v_4$.
		\item Assume now that $v_{i,4}$ and $v_{j,4}$ are on the same SP from
			$u_4$ to $v_4$ but $v_{i,j,4}$ is on the other SP from $u_4$ to
			$v_4$ (by hypothesis there are only two SPs from $u_4$ to $v_4$).
			Since what we prove in the previous point must be true for all
			choices of $i$ and $j$, we have that all nodes $v_{h,4}$, $1\le h\le
			3$, must be on the same SP from $u_4$ to $v_4$, and all nodes in the
			form $v_{i,j,4}$, $1\le i< j\le 3$ must be on the other SP from
			$u_4$ to $v_4$. Consider now these three nodes, $v_{1,2,4}$,
			$v_{1,3,4}$, and $v_{2,3,4}$ and consider their ordering along the
			SP from $u_4$ to $v_4$ that they lay on. No matter what the ordering
			is, there is an index $h\in\{1,2,3\}$ such that the shortest path
			$p_h$ must go through the extreme two nodes in the ordering but not
			through the middle one. But this would contradict fact
			$\mathsf{F}_1$, so it is impossible that we have $v_{i,4}$ and
			$v_{j,4}$ on the same SP from $u_4$ to $v_4$ but $v_{i,j,4}$ is on
			the other SP, for any choice of $i$ and $j$.
	\end{itemize}

	We showed that the nodes $v_{i,4}$ and $v_{j,4}$ can not be on the same
	SP from $u_4$ to $v_4$. But this is true for any choice of the unordered
	pair $(i,j)$ and there are three such choices, but only two SPs from $u_4$
	to $v_4$, so it is impossible to accommodate all the constraints requiring
	$v_{i,4}$ and $v_{j,4}$ to be on different SPs from $u_4$ to $v_4$. Hence we
	reach a contradiction and $B$ can not be shattered.
\end{proof}

The following lemma shows that the bound in~\Cref{lem:unique} is tight.

\begin{lemma}\label{lem:uniquetight}
	There is an undirected graph $G=(V,E)$ such that there is a set $\{(u_i,
	v_i), u_i,v_i\in V, u_i\neq v_i, 1\le i\le 4\}$ with
	$|\SP_{u_1,v_1}|=|\SP_{u_2,v_2}|=2$ and $|\SP_{u_3,v_3}|=|\SP_{u_4,v_4}|=1$
	that is shattered.
\end{lemma}

\begin{figure}[htb]
\centering
\begin{tikzpicture}[scale=0.9]
\GraphInit[vstyle=Classic]
\tikzset{VertexStyle/.append style = { minimum size = 2 pt }}
\Vertex[Lpos=-90]{0}
\Vertex[x=1,y=0,Lpos=-90]{1}
\Vertex[x=2,y=0,Lpos=-90]{2}
\Vertex[x=3,y=0,Lpos=-90]{3}
\Vertex[x=4,y=0,Lpos=-90]{4}
\Vertex[x=5,y=0,Lpos=-90]{5}
\Vertex[x=6,y=0,Lpos=-90]{10}
\Vertex[x=7,y=0]{11}
\Vertex[x=2,y=1,Lpos=-180]{22}
\Vertex[x=2,y=2,Lpos=-180]{21}
\Vertex[x=2,y=3,Lpos=-180]{35}
\Vertex[x=2,y=4,Lpos=-180]{20}
\Vertex[x=2,y=5,Lpos=-180]{36}
\Vertex[x=2,y=6,Lpos=-180]{27}
\Vertex[x=2,y=7,Lpos=-180]{26}
\Vertex[x=2,y=8,Lpos=90]{25}
\Vertex[x=1,y=8,Lpos=90]{24}
\Vertex[x=0,y=8,Lpos=90]{23}
\Vertex[x=3,y=4,Lpos=90]{19}
\Vertex[x=3,y=6,Lpos=90]{37}
\Vertex[x=3,y=8,Lpos=90]{28}
\Vertex[x=4,y=3]{7}
\Vertex[x=4,y=6,Lpos=90]{34}
\Vertex[x=4,y=8,Lpos=90]{39}
\Vertex[x=5,y=1]{6}
\Vertex[x=5,y=5,Lpos=-90]{31}
\Vertex[x=5,y=6]{38}
\Vertex[x=5,y=8,Lpos=90]{33}
\Vertex[x=6,y=1]{9}
\Vertex[x=6,y=2]{8}
\Vertex[x=6,y=4,Lpos=-90]{18}
\Vertex[x=6,y=5]{29}
\Vertex[x=6,y=6]{30}
\Vertex[x=6,y=7]{32}
\Vertex[x=6,y=8]{40}
\Vertex[x=7,y=1]{12}
\Vertex[x=7,y=2]{13}
\Vertex[x=7,y=3]{14}
\Vertex[x=7,y=4]{15}
\Vertex[x=8,y=3]{16}
\Vertex[x=8,y=5]{17}
\Edge(0)(1)
\Edge(1)(2)
\Edge(2)(3)
\Edge(2)(22)
\Edge(3)(4)
\Edge(4)(5)
\Edge(5)(6)
\Edge(6)(7)
\Edge(7)(19)
\Edge(22)(21)
\Edge(21)(8)
\Edge(21)(35)
\Edge(8)(9)
\Edge(9)(10)
\Edge(10)(11)
\Edge(11)(12)
\Edge(12)(13)
\Edge(13)(14)
\Edge(14)(15)
\Edge(15)(18)
\Edge(15)(16)
\Edge(15)(17)
\Edge(18)(19)
\Edge(19)(20)
\Edge(35)(20)
\Edge(20)(36)
\Edge(36)(27)
\Edge(27)(26)
\Edge(27)(37)
\Edge(37)(34)
\Edge(22)(26)
\Edge(26)(25)
\Edge(25)(24)
\Edge(24)(23)
\Edge(25)(28)
\Edge(28)(39)
\Edge(39)(33)
\Edge(33)(40)
\Edge(40)(32)
\Edge(32)(30)
\Edge(30)(29)
\Edge(29)(18)
\Edge(29)(31)
\Edge(31)(38)
\Edge(38)(34)
\end{tikzpicture}
\caption{Graph for~\Cref{lem:uniquetight}}
\label{fig:fourshatter}
\end{figure}

 \begin{proof}
	Consider the undirected graph $G=(V,E)$ in~\Cref{fig:fourshatter}.
	There is a single SP from $0$ to $16$:
	\[
		0, 1, 2, 22, 21, 35, 20, 19, 18, 15, 16\enspace.
	\]
	There is a single SP from $23$ to $17$:
	\[
		23, 24, 25, 26, 27, 36, 20, 19, 18, 15, 17\enspace.
	\]
	There are exactly two SPs from $5$ to $33$:
	\begin{align*}
		&5, 4, 3, 2, 22, 26, 25, 28, 39, 33 \mbox{ and} \\
		&5, 6, 7, 18, 18, 29, 30, 32, 40, 33\enspace.
	\end{align*}
	There are exactly two SPs from $11$ to $34$:
	\begin{align*}
		11, 10, 9, 8, 21, 22, 26, 27, 37, 34 \mbox{ and} \\
		11, 12, 13, 14, 15, 18, 29, 31, 38, 34\enspace.
	\end{align*}
	Let
	$a=((0,16),1)$, $b=((23,17),1)$, $c=((5,33),1/2)$, and $d=((11,34), 1/2)$.
	We can shatter the set $Q=\{a,b,c,d\}$, as shown in~\Cref{tab:uniquetight}.
 \end{proof}

 \begin{table}[htb]
  \centering
    \begin{tabular}{cc}
		\toprule
		$P\subseteq Q$ & Vertex $v$ such that $P=Q\cap R_v$ \\
		\midrule
		$\emptyset$ & 0\\
		$\{a\}$ & 1\\
		$\{b\}$ & 24\\
		$\{c\}$ & 40\\
		$\{d\}$ & 38\\
		$\{a,b\}$ & 20\\
		$\{a,c\}$ & 2\\
		$\{a,d\}$ & 21\\
		$\{b,c\}$ & 25\\
		$\{b,d\}$ & 27\\
		$\{c,d\}$ & 29\\
		$\{a,b,c\}$ & 19\\
		$\{a,b,d\}$ & 15\\
		$\{a,c,d\}$ & 22\\
		$\{b,c,d\}$ & 26\\
		$\{a,b,c,d\}$ & 18\\
		\bottomrule
	\end{tabular}
	\caption{How to shatter $Q=\{a,b,c,d\}$ from~\Cref{lem:uniquetight}.}
	\label{tab:uniquetight}
 \end{table}

% XXX I cannot prove the following, but it should be possible.
%\begin{lemma}
%	There exists no undirected graph $G=(V,E)$ such that it is possible to
%	shatter a set
%	\[
%		B=\{((u_i, v_i),x_i), 1\le i\le 5\}\subseteq\domain\times[0,1]
%	\]
%	if
%	\[
%		\sum_{i=1}^5 \sigma_{u_iv_i}<9\enspace.
%	\]
%\end{lemma}

We pose the following conjecture, which would allow us to
generalize~\Cref{lem:unique}, and develop an additional stopping rule for
\staticalgo based on the empirical pseudodimension.
\begin{conjecture}
	Given $n>0$, there exists no undirected graph $G=(V,E)$ such that it is
	possible to shatter a set
	\[
		B=\{((u_i, v_i),x_i), 1\le i\le n\}\subseteq\domain\times[0,1]
	\]
	if
	\[
		\sum_{i=1}^n \sigma_{u_iv_i}|<\binom{n}{\lfloor n/2\rfloor}\enspace.
	\]
\end{conjecture}

\section{Additional Experimental Results}\label{app:exper}
In this section we show additional experimental results, mostly limited to
additional figures like \Cref{fig:schedules,fig:errors} but for other graphs.
The figures we present here exhibits the exact same behavior as those
in \Cref{sec:exper}, and that is why we did not include include them in the main
text. Figures corresponding to \Cref{fig:schedules} are shown in
\Cref{fig:appschedules} and those corresponding to \Cref{fig:errors} are shown in
\Cref{fig:apperrors}.

\begin{figure*}[ht]
	\centering
	\begin{subfigure}[b]{.325\linewidth}
		\centering
		\includegraphics[width=\linewidth]{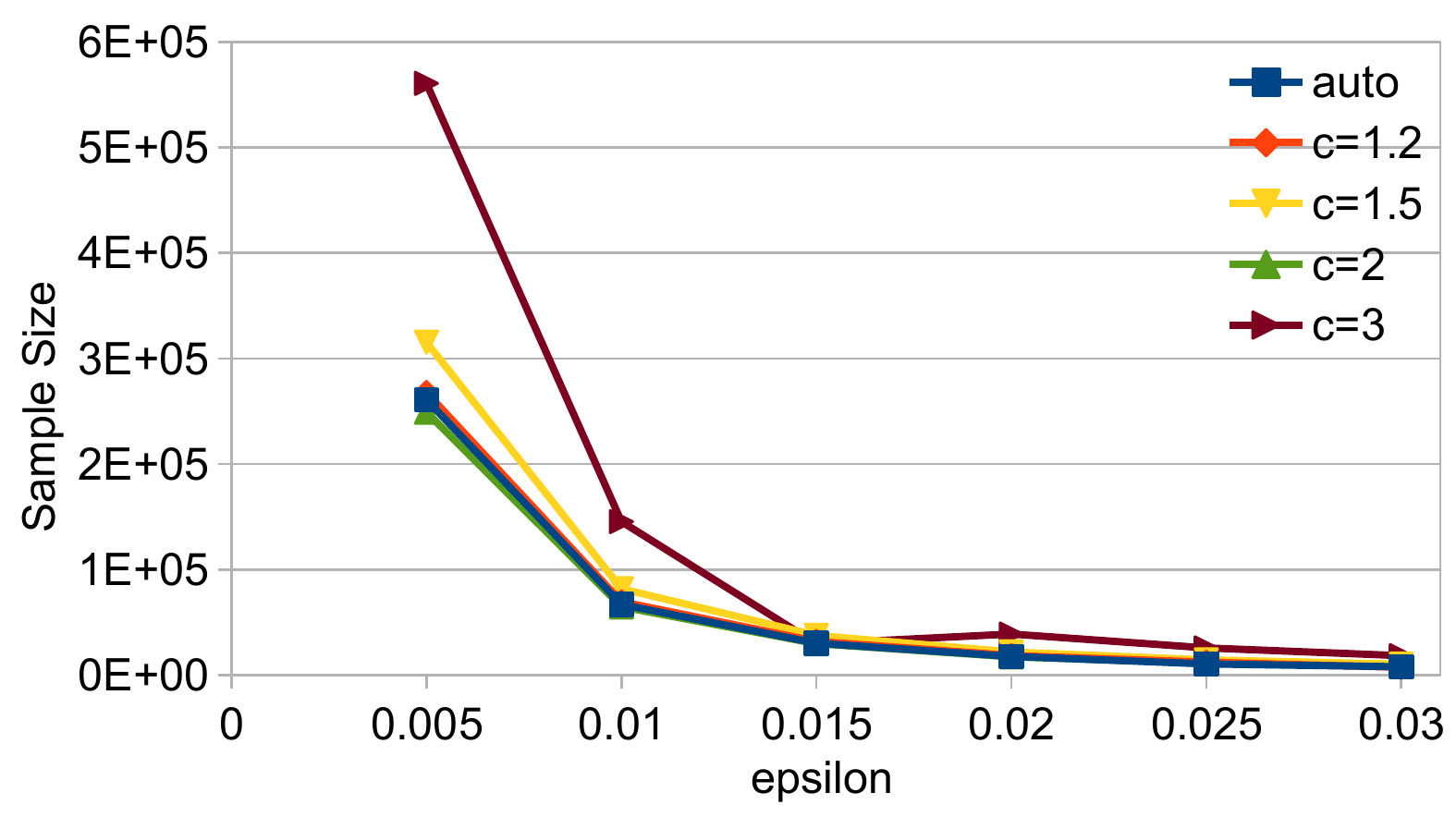}
		\caption{Email-Enron}\label{fig:enronschedules}
	\end{subfigure}
	\begin{subfigure}[b]{.325\linewidth}
		\centering
		\includegraphics[width=\linewidth]{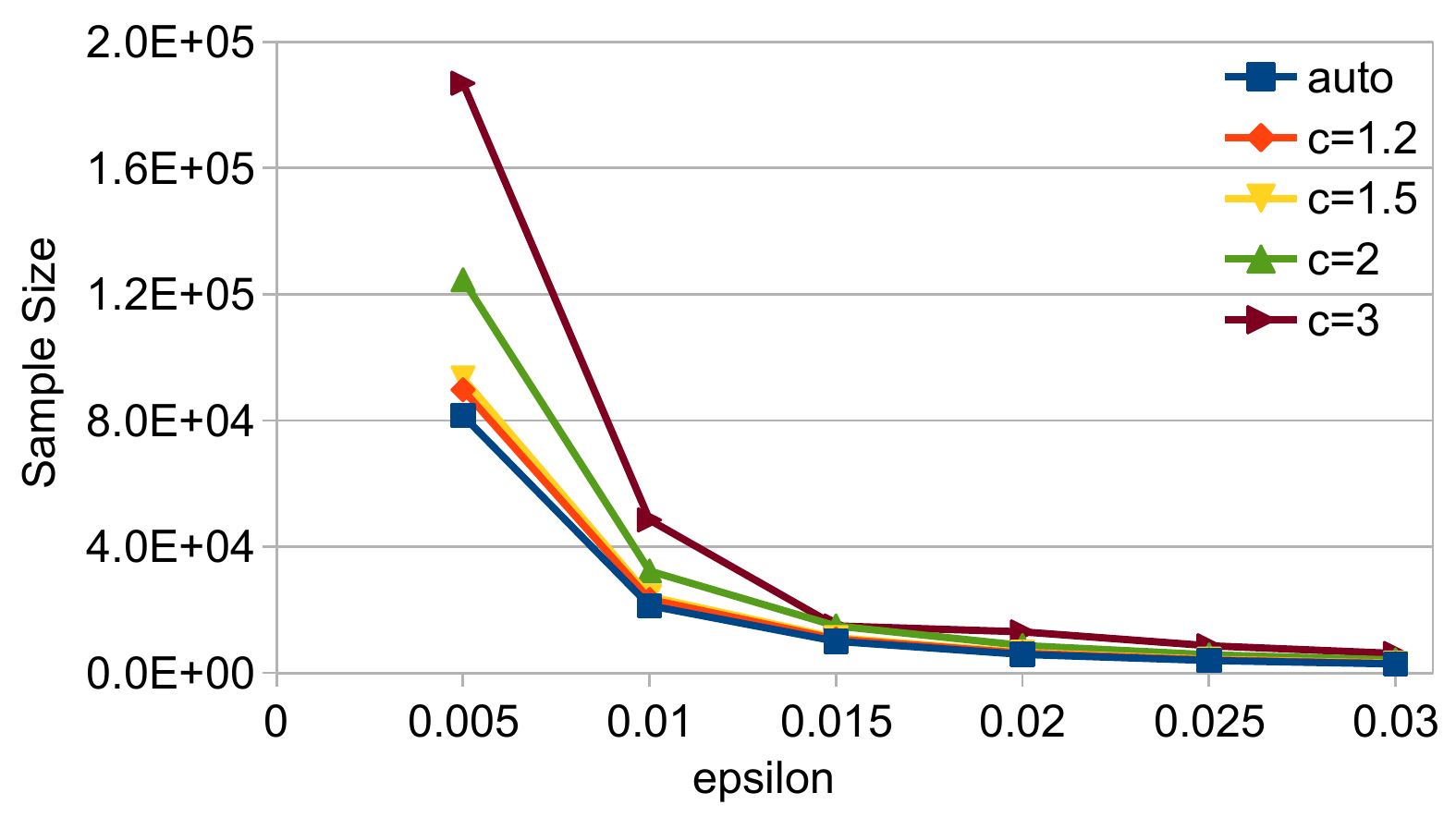}
		\caption{Soc-Epinions1}\label{fig:epinionsschedules}
	\end{subfigure}
	\begin{subfigure}[b]{.325\linewidth}
		\centering
		\includegraphics[width=\linewidth]{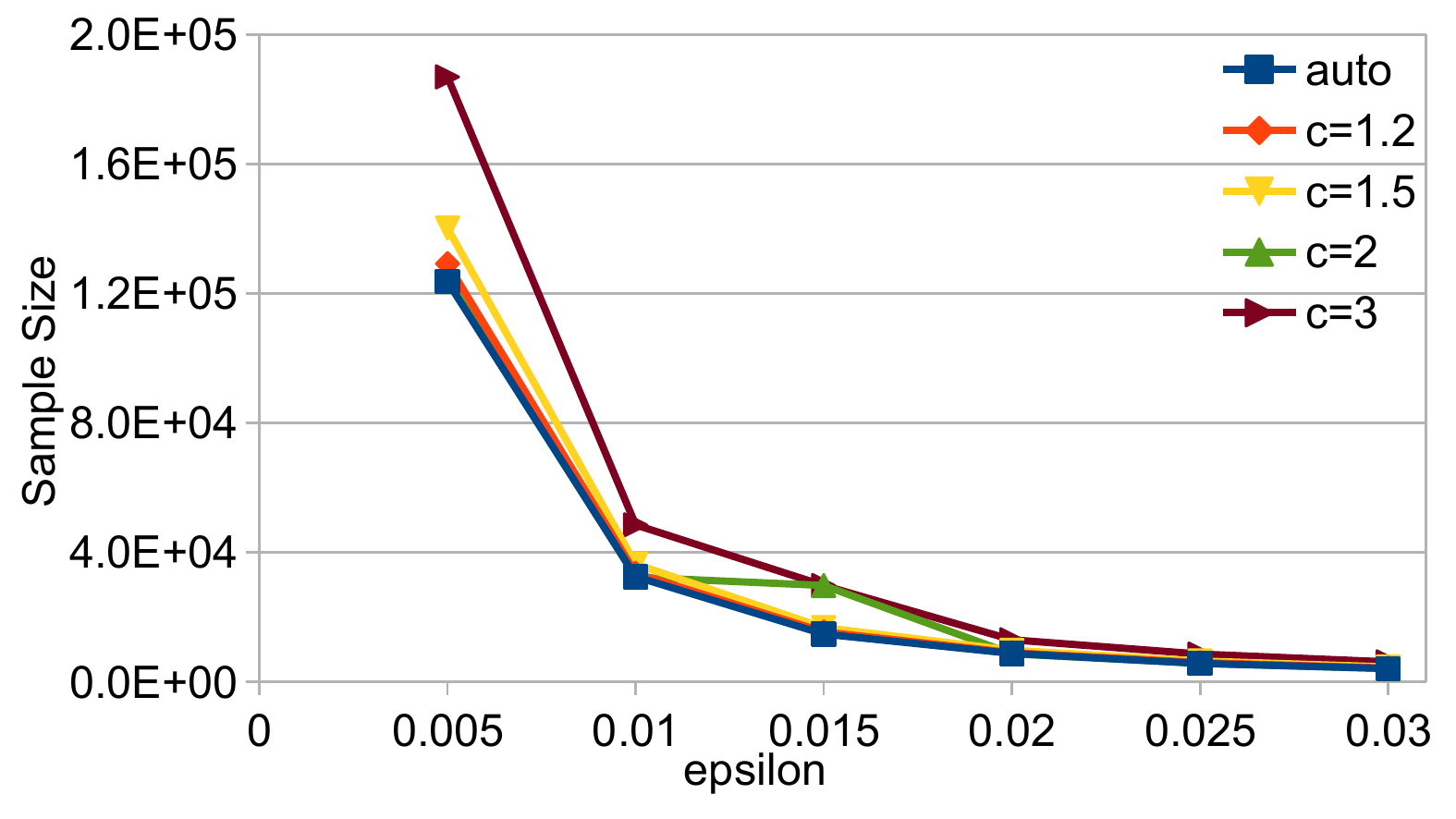}
		\caption{Cit-HepPh}\label{fig:hepphschedules}
	\end{subfigure}
	\caption{Final sample size for different sample schedules}
	\label{fig:appschedules}
\end{figure*}

\begin{figure*}[ht]
	\centering
	\begin{subfigure}[b]{.325\linewidth}
		\centering
		\includegraphics[width=\linewidth]{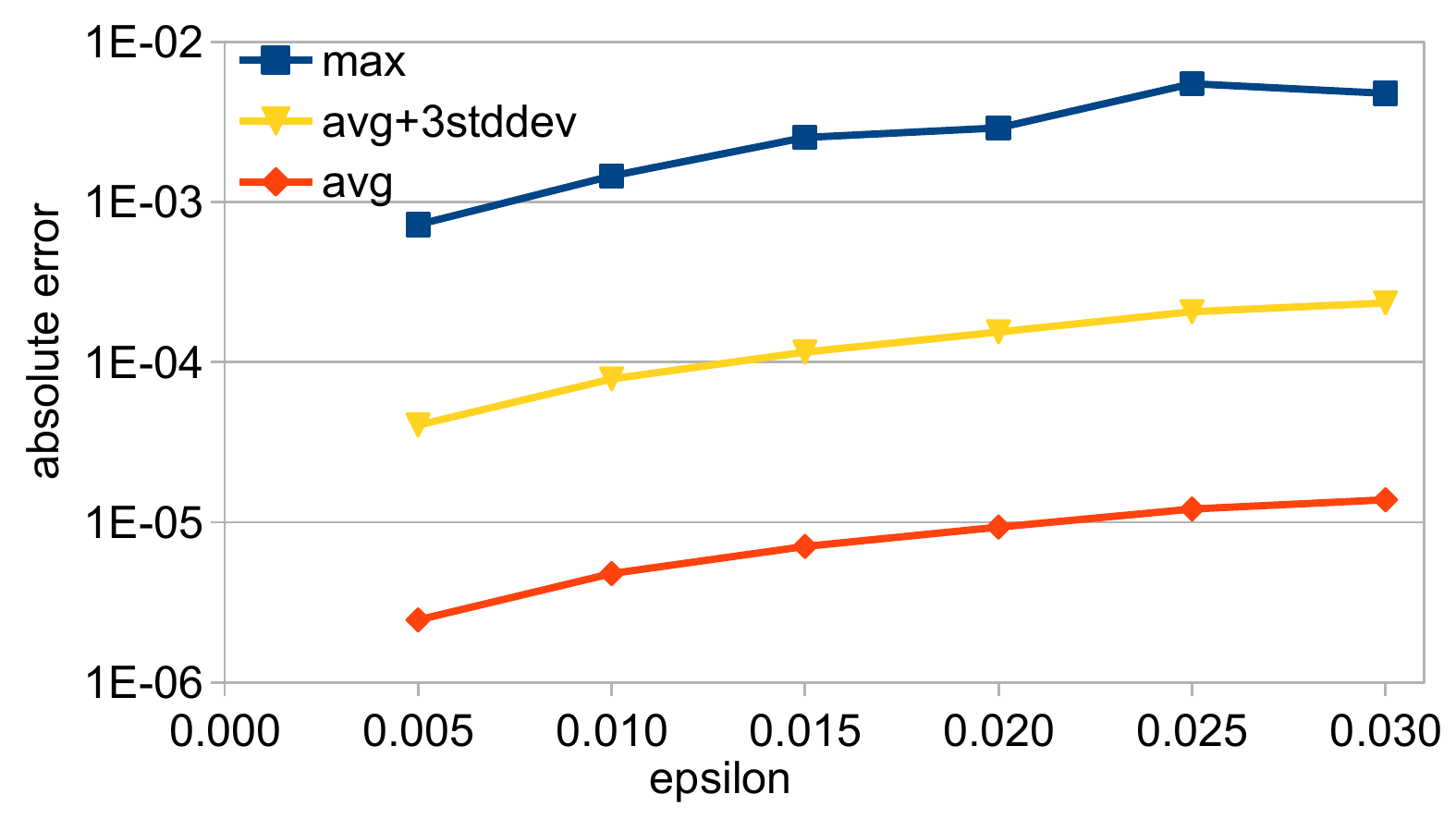}
		\caption{Email-Enron}\label{fig:enronerror}
	\end{subfigure}
	\begin{subfigure}[b]{.325\linewidth}
		\centering
		\includegraphics[width=\linewidth]{epinions-error.pdf}
		\caption{Soc-Epinions1}\label{fig:epinionserror}
	\end{subfigure}
	\begin{subfigure}[b]{.325\linewidth}
		\centering
		\includegraphics[width=\linewidth]{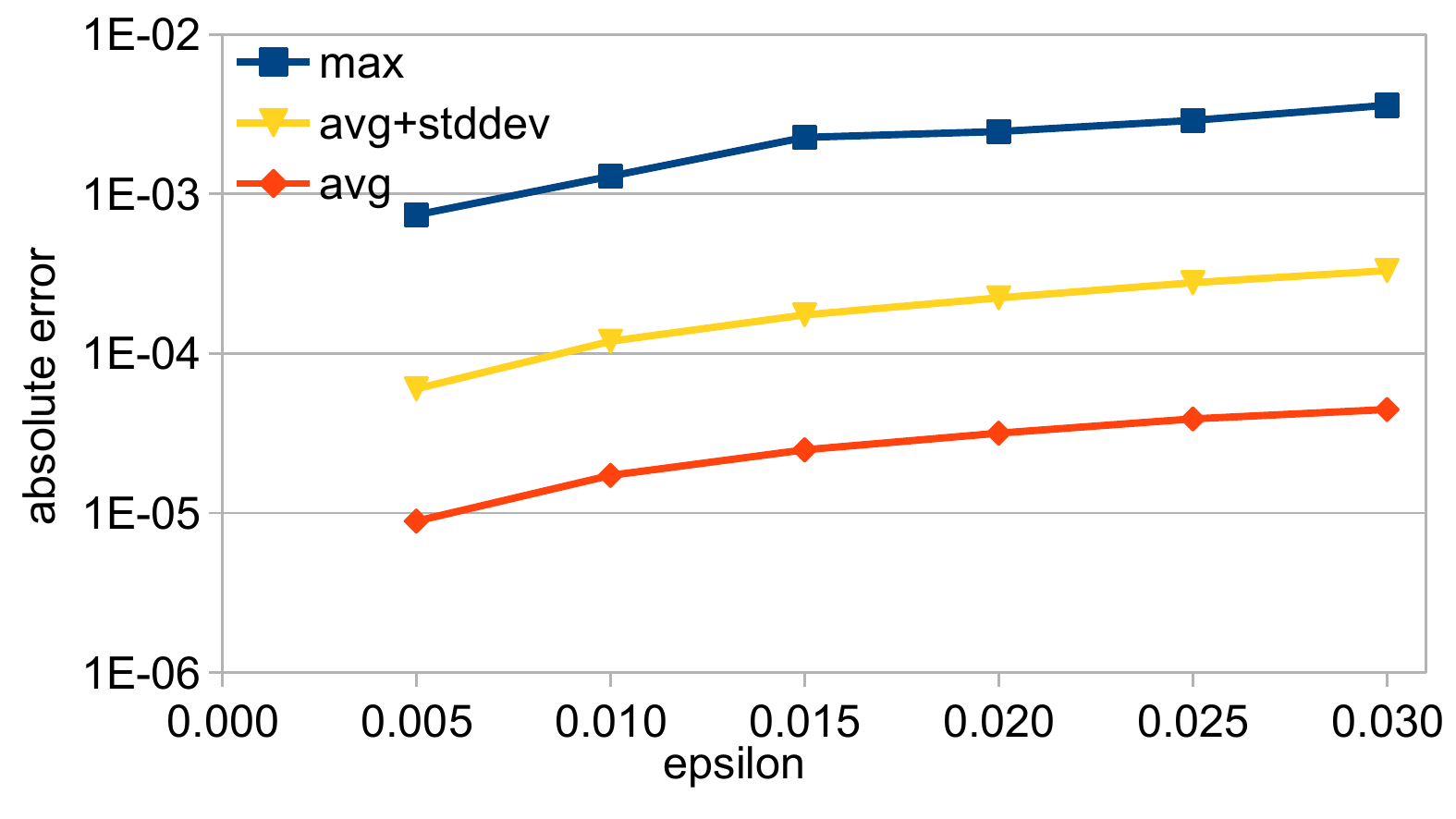}
		\caption{Cit-HepPh}\label{fig:heppherror}
	\end{subfigure}
	\caption{Absolute error evaluation}
	\label{fig:apperrors}
\end{figure*}

\section{Relative-error Rademacher Averages}\label{app:raderel}
\newcommand*\prob{\pi}

In this note we show how to obtain \emph{relative
$(p,\varepsilon)$-approximations} as defined by~\citet{HarPS11} (see
Def.~\ref{def:peapprox}) using a relative-error variant of the Rademacher
averages.

\subsection{Definitions}
Let $\domain$ be some domain, and $\family$ be a family of functions from
$\domain$ to $[a,b]$, an interval of the non-negative reals.\footnote{We
conjecture that the restriction to the non-negative reals can be easily removed.}
Assume that $\prob$ is a probability distribution on $\domain$. For any
$f\in\family$, let $\expectation_\prob[f]$ be the expected value of $f$
w.r.t.~$\prob$. Let $A=\{a_1,\dotsc,a_n\}$ be a collection of $n$ elements of
$\domain$. For any $f\in\family$, let
\[
	\tilde{f}(A)=\frac{1}{n}\sum_{i=1}^n f(a_i)\enspace.
\]

\begin{definition}[\citep{HarPS11}]\label{def:peapprox}
	Given $p\in(0,1)$ and $\varepsilon\in(0,1)$, a \emph{relative
	$(p,\varepsilon)$-approximation} for $\family$ is a collection $A$ of
	elements of $Z$ such that
	\begin{equation}\label{eq:peapprox}
		\sup_{f\in\family}\frac{|\expectation_\prob[f]-\tilde{f}(A)|}{\max\{p,\expectation_\prob[f]\}}\le\varepsilon\enspace.
	\end{equation}
\end{definition}

\paragraph*{Fixed-sample bound} \citet{HarPS11} showed that, when the functions
of $\family$ only take values in $\{0,1\}$ and $\family$ has finite
VC-dimension, then a sufficiently large collection $\Sam$ of $n$ elements of
$\domain$ sampled independently according to $\prob$ is a relative
$(p,\varepsilon)$-approximation for $\family$ with probability at least
$1-\delta$, for $\delta\in(0,1)$.

\begin{theorem}[Thm.~2.11~\citep{HarPS11}]
	Let $\family$ be a family of functions from $\domain$ to $\{0,1\}$, and let
	$d$ be the VC-dimension of $\family$. Given $p,\varepsilon,\delta\in(0,1)$, let
	\begin{equation}\label{eq:HarPSsamplesize}
		n=O\left(\frac{1}{\varepsilon^2p}\left(d\log\frac{1}{p}+\ln\frac{1}{\delta}\right)\right),
	\end{equation}
	and let $\Sam$ be a collection of $n$ elements of $\domain$ sampled
	independently according to $\prob$. Then,
	\[
		\Pr\left(\sup_{f\in\family}\frac{|\expectation_\prob[f]-\tilde{f}(A)|}{\max\{p,\expectation_\prob[f]\}}>\varepsilon\right)<\delta,
	\]
	or, in other words, $\Sam$ is a relative $(p,\varepsilon)$-approximation for
	$\family$ with probability at least $1-\delta$.
\end{theorem}

\paragraph*{Related works} The bound in~\eqref{eq:HarPSsamplesize} is an
extension of a result by~\citet{LiLS01} obtained for families of
\emph{real-valued} functions taking values in $[0,1]$, and using the
\emph{pseudodimension} of the family instead of the VC-dimension. The
original result by~\citet{LiLS01} shows how large should $\Sam$ be in order for
the quantity
\begin{equation}\label{eq:LiLSapprox}
	\sup_{f\in\family}\frac{|\expectation_\prob[f]-\tilde{f}(\Sam)|}{\expectation_\prob[f]+\tilde{f}(\Sam)+p}
\end{equation}
to be at most $\varepsilon$ with probability at least $1-\delta$. Some constant
factors are lost in the adaptation of the measure from~\eqref{eq:LiLSapprox} to
the one on the l.h.s.~of~\eqref{eq:peapprox}. The quantity
in~\eqref{eq:LiLSapprox} has been studied often in the literature of
statistical learning theory, see for example~\citep[Sect.~5.5]{AnthonyB99},
\citep[Sect.~5.1]{BoucheronBL05}, and \citep{Haussler92}, while other works
(e.g.,~\citep[Sect.~5.1]{BoucheronBL05}, \citep{CortesGM13},
\citep{AnthonyST93}, and \citep{BartlettL99}) focused on the quantity
\[
	\sup_{f\in\family}\frac{|\expectation_\prob[f]-\tilde{f}(\Sam)|}{\sqrt{\expectation_\prob[f]}}\enspace.
\]

\subsection{\texorpdfstring{Obtaining a relative
$(p,\varepsilon)$-approximation}{Obtaining a relative (p,epsilon)-approximation}}
In this note we study how to bound the quantity on the
l.h.s.~of~\eqref{eq:peapprox} directly, without going through the quantity
in~\eqref{eq:LiLSapprox}. By following the same steps
as~\citep[Thm.~2.9(ii)]{HarPS11}, we can extend our results to the quantity
in~\eqref{eq:LiLSapprox}. The advantage of tackling the problem directly is that
we can derive \emph{sample-dependent bounds} with \emph{explicit} constants.
Moreover, the use of (a variant of) Rademacher averages allows us to obtain
stricter bounds to the sample size.

Let $\Sam=\{X_1,\dotsc,X_n\}$ be a collection of $n$ elements from $\domain$
sampled independently according to $\prob$. Let $\sigma_1,\dotsc,\sigma_n$ be
independent Rademacher random variables $\sigma_1,\dotsc,\sigma_n$, independent
from the samples. Consider now the random variable
\[
	\mathsf{r}\rade(\family,\Sam,p)=\expectation_\sigma\left[
		\sup_{f\in\family}\frac{1}{n\max\{p,\expectation_\prob[f]\}}\sum_{i=1}^n\sigma_if(X_i)\right],
\]
which we call the \emph{conditional $p$-relative Rademacher average of $\family$ on
$\Sam$}. We have the following result connecting this quantity to the
$(p,\varepsilon)$ approximation condition.

\begin{theorem}\label{thm:raderel}
	Let $\Sam$ be a collection of $n$ elements of $\domain$ sampled
	independently according to $\prob$. With probability at least $1-\delta$,
	\[
		\sup_{f\in\family}\frac{|\expectation_\prob[f]-\tilde{f}(\Sam)|}{\max\{p,\expectation_\prob[f]\}}\le
		2 \mathsf{r}\rade(\Sam,\family,p) + 3\frac{|b-a|}{p}\sqrt{\frac{\ln(2/\delta)}{n}}\enspace.
	\]
\end{theorem}

The proof of~\Cref{thm:raderel} follows step by step the proof
of~\Cref{thm:supdevbound} (\citep[Thm.~26.4]{ShalevSBD14}), with the only
important difference that we need to show that the quantities
\[
	\sup_{f\in\family}\frac{\expectation_\prob[f]-\tilde{f}(\Sam)}{\max\{p,\expectation_\prob[f]\}}
\]
and $\mathsf{r}\rade(\family,\Sam,p)$, seen as functions of
$\Sam=\{X_1,\dotsc,X_n\}$, satisfy the \emph{bounded difference inequality}.

\begin{definition}[Bounded difference inequality]
	Let $g : \mathcal{X}^n \rightarrow \mathbb{R}$ be a function of $n$
	variables. The function $g$ is said to satisfy the bounded difference
	inequality iff for each $i$, $1\le i\le n$ there is a nonnegative constant
	$c_i$ such that:
	\begin{equation}\label{eq:bounded}
		\sup_{\substack{x_1,\dotsc,x_n\\x_i'\in\mathcal{X}}}|g(x_1,\dotsc,x_n)-g(x_1,\dotsc,x_{i-1},x'_i,x_{i+1},\dotsc,x_n)|\le
			c_i\enspace.
	\end{equation}
\end{definition}

We have the following results, showing that indeed the quantities above satisfy
the bounded difference inequality.

\begin{lemma}\label{lem:boundedsup}
	The function
	\[
		g(X_1,\dotsc,X_n)=\sup_{f\in\family}\frac{\expectation_\prob[f]-\tilde{f}(\Sam)}{\max\{p,\expectation_\prob[f]\}}
	\]
	satisfies the bounded difference inequality~\eqref{eq:bounded} with
	constants
	\[
		c_i=\frac{|b-a|}{np}\enspace.
	\]
\end{lemma}

%Before proving this result, we need to introduce the following known
%fact.\footnote{The fact is known, but we can not find a reference for it, so we
%report the proof for completeness.}
%
%\begin{fact}\label{fact:supdev}
%	Let $\mathcal{X}$ be a domain and let $g$ and $h$ be two function from
%	$\mathcal{X}$ to $\mathbb{R}$. Then
%	\[
%		\sup_{x\in\mathcal{X}}g(x) - \sup_{x\in\mathcal{X}}h(x) \le
%		\sup_{x\in\mathcal{X}}(g(x)-h(x))\enspace.
%	\]
%\end{fact}
%\begin{proof}
%	\todo{Write.}
%\end{proof}

\begin{proof}[of~\Cref{lem:boundedsup}]
	Let $\Sam=\{X_1,\dotsc,X_n\}$ and, for any $i$, $1\le i\le n$, let
	\[
		\Sam'_i=\{X_1,\dotsc,X_{i-1},X'_i,X_{i+1},\dotsc,X_n\},
	\]
	i.e., we replaced the random variable $X_i$ with another random variable
	$X'_i$, sampled independently according to the same distribution.
	For any function $f\in\family$ let
	\[
		\phi_f(\Sam)=\frac{\expectation_\prob[f]-\tilde{f}(\Sam)}{\max\{p,\expectation_\prob[f]\}}\enspace.
	\]
	It is easy to see that
	\begin{equation}\label{eq:boundedphi}
		|\phi_f(\Sam) -\phi_f(\Sam'_i)|\le \frac{|b-a|}{np}
	\end{equation}

	We have
	\begin{align}\label{eq:equivsup}
		&|g(X_1,\dotsc,X_n)-g(X_1,\dotsc,X_{i-1},X'_i,X_{i+1},\dotsc,X_n)=\nonumber\\
		&|g(\Sam'_i)-g(\Sam)| = \nonumber\\
		&\left|\sup_{f\in F}\phi_f(\Sam'_i) - \sup_{f\in
		F}\phi_f(\Sam)\right|\enspace.
	\end{align}
	To simplify the notation, let now $\ell\in\family$ denote one of the
	functions for which the supremum is attained on $\Sam'_i$, and let
	$h\in\family$ be on of the functions for which the supremum is attained on
	$\Sam$. Then we can rewrite~\eqref{eq:equivsup} as
	\[
		|\phi_\ell(\Sam'_i) - \phi_h(\Sam)|\enspace.
	\]
	Assume w.l.o.g.~that
	\begin{equation}\label{eq:assumptionsup}
		\phi_h(\Sam)\le\phi_\ell(\Sam'_i),
	\end{equation}
	(the other case follows by symmetry). We have
	\begin{equation}\label{eq:factsup}
		\phi_\ell(\Sam) \le \phi_h(\Sam)
	\end{equation}
	because $h$ attains the supremum over all possible
	$f\in F$ on $X_1,\dotsc,X_n$. This and our
	assumption~\eqref{eq:assumptionsup} imply that it must be
	\[
		\phi_\ell(\Sam)\le
		\phi_\ell(\Sam'_i)\enspace.
	\]
	From this and~\eqref{eq:boundedphi} we have
	\[
		\phi_\ell(\Sam'_i)\le
		\phi_\ell(\Sam) + \frac{|b-a|}{np}\enspace.
	\]
	Then from this and from~\eqref{eq:factsup} we have
	\begin{align*}
		\phi_\ell(\Sam'_i) - \phi_h(\Sam)\le \left(\phi_\ell(\Sam)+\frac{|b-a|}{np}\right) -
		\phi_\ell(\Sam)\le \frac{|b-a|}{np}\enspace.
	\end{align*}
\end{proof}

Using the same steps as the above proof, we can prove the following result about
the conditional $p$-relative Rademacher average.

\begin{lemma}\label{lem:boundedrade}
	The function
	\[
		g(X_1,\dotsc,X_n)=\mathsf{r}\rade(\{X_1,\dotsc,X_n\},\family,p)
	\]
	satisfies the bounded difference inequality~\eqref{eq:bounded} with
	constants
	\[
		c_i=\frac{|b-a|}{np}\enspace.
	\]
\end{lemma}

The following result is the analogous of~\Cref{thm:radeboundw}
(\citep[Thm.~3]{RiondatoU15}) for the conditional $p$-relative Rademacher averages.

\begin{theorem}\label{thm:radeboundwr}
	Let $\functionw_\mathsf{r}: \mathbb{R}^+ \to \mathbb{R}^+$ be the function
  \begin{equation}\label{eq:functionwr}
	\functionw_\mathsf{r}(s) =
	\frac{1}{s}\ln\displaystyle\sum_{\mathbf{v}\in
	\mathcal{V}_\Sam}\mathrm{exp}(s^2\|\mathbf{v}\|^2/p(2\ell^2)),
  \end{equation}
  where $\|\cdot\|$ denotes the Euclidean norm. Then
  \begin{equation}\label{eq:functionrwbound}
	\mathsf{r}\rade(\family,\Sam,p)\leq\min_{s\in\mathbb{R}^+}\functionw_\mathsf{r}(s)\enspace.
\end{equation}
\end{theorem}

The proof follows the same steps as the one for~\citep[Thm.~3]{RiondatoU15},
with the additional initial observation that
\[
	\mathsf{r}\rade(\Sam,\family,p)\le\frac{1}{p}\rade(\Sam,\family)\enspace.
\]

\fi

\end{document}